\documentclass[journal,10pt,twocolumn]{IEEEtran}
\usepackage{epsfig}
\usepackage{graphicx}
\usepackage{subfigure}
\usepackage{amssymb}
\usepackage{makeidx}
\usepackage{amsmath}
\usepackage{amsthm}
\usepackage{graphicx}
\usepackage{epsf}
\usepackage{epsfig}
\usepackage{psfig}
\usepackage{ccaption}
\usepackage{array}
\usepackage{tabularx}
\usepackage{multirow}
\usepackage{epsfig}
\usepackage{cite}
\usepackage{procedure,algorithm,algorithmic}

\newtheorem{theorem}{Theorem}

\newtheorem{lemma}{Lemma}

\newtheorem{definition}{Definition}

\newtheorem{remark}{Remark}

\begin{document}

\title{Coordination and Communication of Autonomous Subsystems in Cyber Physical Systems: \\A Mechanism Learning Approach}
\author{Husheng Li
\thanks{H. Li is with the Department of Electrical Engineering and Computer Science, the University of Tennessee, Knoxville, TN (email: husheng@eecs.utk.edu, phone number: 865-974-3861, fax number: 865-974-5483,  address: 310 Ferris Hall, 1508 Middle Drive, Knoxville, TN, 37996). This work is supported by NSF 1525418.}}
\maketitle

\begin{abstract}
In the control of many autonomous subsystems, such as autonomous vehicles or UAV networks, a centralized control may be hindered by the prohibitive complexity, limited communication bandwidth, or private information of subsystems. Therefore, it is desirable for the control center to coordinate the controls of subsystems by designing mechanisms such as pricing, which makes the local optimizations of subsystem dynamics also maximize the reward of the total system, namely the social welfare. The economics framework of mechanism design is employed for the coordination of the autonomous subsystems. To address the challenge of dynamics, which are not considered in conventional economics mechanism design, and the complexity of private information, the approaches of geometrization and machine learning are employed, by endowing different geometric structures to the problem. The theoretical framework is applied in the context of urban aerial mobility, where the numerical simulations show the validity of the proposed framework. 
\end{abstract}

\section{Introduction} In recent year, the autonomous control of systems has received substantial studies, particularly, motivated by the demand of autonomous driving and UAV networks, et al. The control may involve the coordination of many autonomous subsystems; e.g., each vehicle is a subsystem which controls its own dynamics. However, the traditional control theory faces substantial challenges due to the reasons, which makes centralized control difficult:
\begin{itemize}
\item Complexity: The size of the system under control is becoming prohibitively large, which may consists hundreds or even thousands of subsystems (e.g., a large UAV network). One of the reasons is the new generation of communication systems; e.g., the 5G New Radio (NR) can support the wireless connection of millions of devices in the mMTC mode. This complexity brings challenges to the computing, The trend of edge computing desires to distribute the computing to the edges, instead of a centralized computing, thus reducing the burden on the communications and computing center. 

\item Uncertainty: Even if the system provides sufficient communication and computing resources, it may still be impossible for the center to compute the control actions for the subsystems. It is because some subsystems may not share private information with the control center. For example, if each subsystem is an economical agent, its utility function may be kept secret, in order to maximize its own reward. Therefore, the control center needs to coordinate the dynamics of subsystems subject to the uncertainties.
\end{itemize}

We propose to study the coordination of autonomous subsystems in the framework of economics. One of the key problems to be solved in economics is how to design the rules to maximize the social welfare subject to the uncertainty of individual agents, as well as their incentives. For example, in auctions, each agent has its own value on the object, while it may not want to disclose its true value in the bid. The mechanism that the agent with the highest bid wins and is charged at the second highest bid has been shown to be incentive compatible, which makes disclosing the true value the optimal bidding strategy. The generic framework of such economic problems is the theory of mechanism design, which studies how to design efficient mechanism such that the local reward maximization of each agent results in the maximization of social welfare. The most usual mechanism is the pricing of commodities, such as auctions or digital goods. The mechanism design can also be considered as a distributed computing problem, in which the center computes the desired function (e.g., the optimal prices) given the information provided by the agents. Therefore, it is also highly related to the research on communication complexity. 

In this paper, we employ the framework of mechanism design to study the coordination of autonomous control of subsystems, where the private information of each subsystem is the utility function, and the control center computes prices according to the responses from the subsystems. Compared with the standard theory of mechanism design, the following new challenges are incurred in the context of autonomous control:
\begin{itemize}
\item Dynamics: In traditional studies of mechanism design, the setup is static. However, in the autonomous control, the utility functions of different subsystems may change with their own system states. Therefore, the mechanism of coordination needs to be adaptive to the system states. 
\item Complex parameters: In the context of autonomous control, the private information of each subsystem is the utility function. The functional space of the private information makes the design substantially more difficult than that of auctions in which the private information is real number (namely the value of object). 
\end{itemize}

To address the above new challenges, we employ the following approaches for the mechanism of autonomous control:
\begin{itemize}
\item Geometrization: We will cast the mechanism design in a geometric framework, in which the setup is formulated as a vector bundle. Moreover, the geometric structure of mechanism, namely the relationship between distribution, foliation and level set, will be used to guide the mechanism design.
\item Machine learning: The unknown information of each subsystem needs to be learned from their responses to the coordination during the operation. The estimation of the utility function will be accomplished as machine learning using samples and different assumptions on the geometric structure.  
\end{itemize}

\section{Related Works}
\subsection{System of Systems}

\subsection{Game Design for Optimization} In \cite{NaLi2013} and \cite{NaLi2014}, the goal function, determined by the actions of multiple agents, is optimized by designing local games for different agents. The game design is similar to the mechanism design framework in which the desired output is also obtained by designing proper games for different subsystems. However, the work in \cite{NaLi2013} and \cite{NaLi2014} are substantially different from the studies in this paper in the following aspect: (a) Goal function: In \cite{NaLi2013,NaLi2014} the goal function is known to all agents, while in our context the goal function is unknown due to the private information of different sub-systems. (b) Communications: In \cite{NaLi2013,NaLi2014} the agents exchange information for the purpose of optimization, while the sub-systems are not allowed to communicate. (c) Dynamics: The optimization in \cite{NaLi2013,NaLi2014} is stationary, while the sub-systems do not exchange information in the mechanism design. (d) Privacy: In \cite{NaLi2013,NaLi2014}, there is no private information. The reason for the distributed computing is the prohibitively high complexity of the goal function. (e) Incentive: In the context of this paper, the sub-systems are selfish; therefore, the mechanism design needs to be incentive compatible. In a contrast, the distributed optimization in \cite{NaLi2013,NaLi2014} is collaborative in the agents, which simply follow the given rule. 

\subsection{Network Utility Maximization} Essentially the mechanism design is to solve the conflict of demand and supply. In the context of control of autonomous subsystems, the coordination is to resolve the conflict of subsystems subject to limited resources, such as space in the air traffic control. Such a demand-supply relationship, as the focus of studies in economics, has been exploited by F. Kelly in the celebrated theory of network utility maximization (NUM) \cite{Kelly1997,Kelly1998,Kelly2004,Chiang2007}. Although this generic framework applies to any data networks, the networking in NCSs is significantly different from traditional data networks, due to the different sources of communication demand, and the different purposes of supply: (A) \textit{Demand side}: the demand of traditional data network is explicit, since it directly consists of the requirements of throughput and delay; meanwhile that of NCS is implicit, since the controlled physical dynamics do not explicitly demand bits and latency; (B) \textit{Supply side}: In traditional networks, the freshness of data packets either does not change (e.g., in elastic-traffic networks without deadline \cite{Tassiulas1992,Tassiulas1995}), or changes with a binary logic (either catch or miss the deadline \cite{Dua2007}), or changes linearly. The major difference between the theory of NUM and the coordinated control is that the setup in NUM is linear, namely the social wealth function is the sum of the utility functions and the constraints of link capacities are also linear; in a contrast, the coordinated control in this paper handled more generic nonlinear problems, which is much more difficult. 

\subsection{Decentralized Control} In decentralized control \cite{Siljak}, the control actions are taken by the agents, which is the same as the coordinated control. However, it is because of the lack of control center and information collection in the decentralized control. In the coordinated control, there does exist a center; however, it cannot directly inform each agent the control action, since (a) it does not know the utility functions of agents, and thus the optimal control action; (b) it is the agents that make the decision. In the decentralized control, the uncertainty is the lack of global state information, while it is the unknown private information at each sub-system in the coordinated control. Moreover, it is assumed that each agent is selfless in the decentralized control, while each sub-system is selfish and aims at maximizing its own reward in the coordinated control. The design of many decentralized control is analytic, while that of the coordinated control is data driven. These differences are summarized in Fig. \ref{fig:radar}.

\begin{figure}
  \centering
  \includegraphics[scale=0.55]{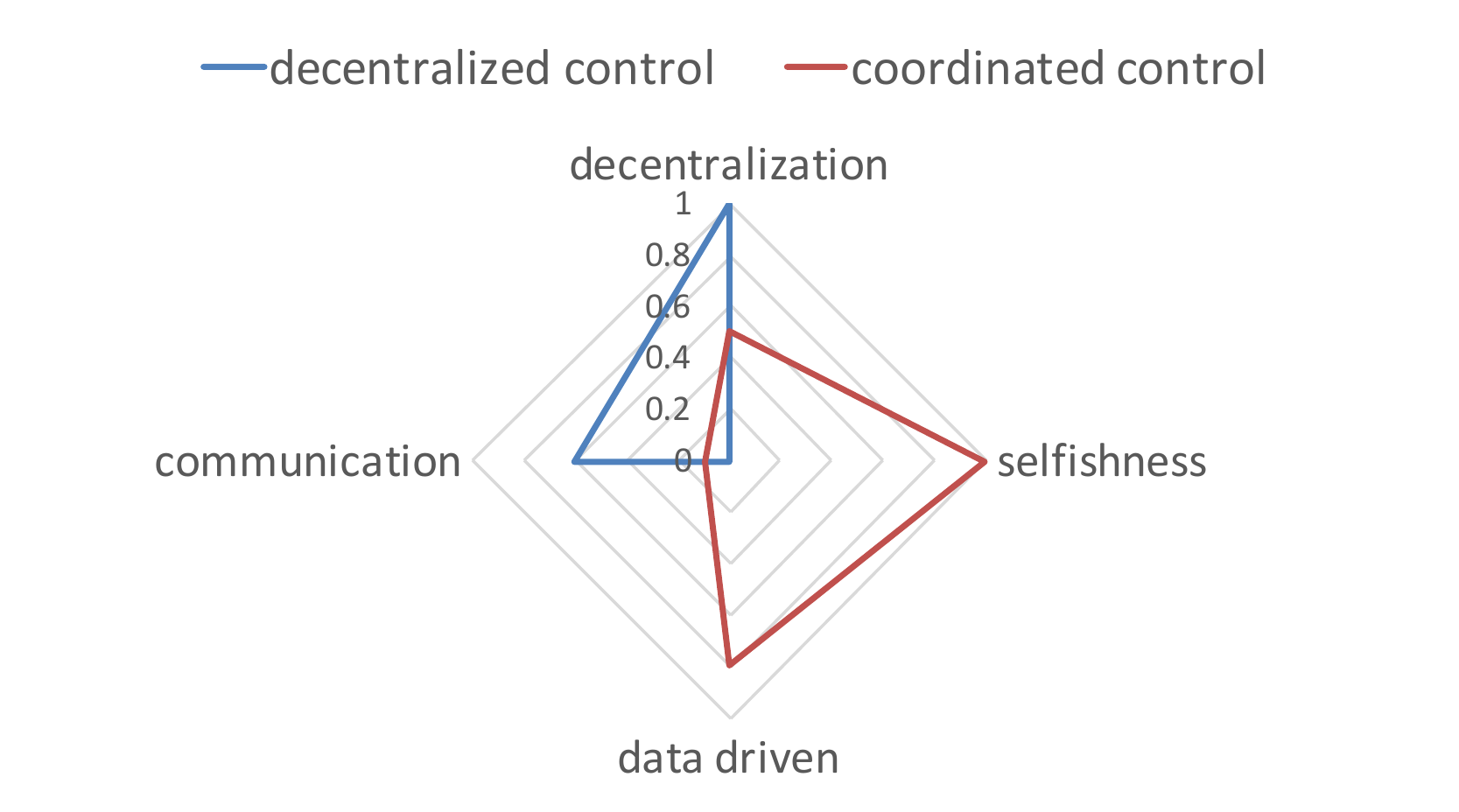}
  \caption{Comparison between the decentralized control and coordinated control.}\label{fig:radar}
  \vspace{-0.1in}
\end{figure}

\subsection{Mechanism Learning} The study on mechanism design was originated by R. Myerson. A set-theoretic approach was proposed in \cite{Hurwicz2006} for mechanism design; the counterpart for the continuous value case was described in \cite{Hurwicz2006}, by exploiting calculus over manifolds. In \cite{Vohra2011}, the linear programming approach was leveraged for the mechanism design, with the emphasis on the incentive compatibility. In \cite{Borgers2014}, a comprehensive introduction to mechanism design is given. 
In traditional studies, the mechanism is designed using explicit analysis (e.g., linear programming \cite{Vohra2011} and level sets \cite{Hurwicz2006}) by human researchers. However, except for the simple case of two agents and a single object \cite{Vohra2011}, the optimal solutions to most mechanism design problems have not been identified, probably due to the high complexity of the problem. In the last two decades, there has been a trend to design the mechanism using machine learning, which takes the numerical methodology based on samples and is coined \textit{automated mechanism design} (AMD) \cite{Conitzer2002}. Thanks to the rapidly increasing computational capabilities of modern computers, there have been substantial breakthroughs in the area of AMD. Essentially, the AMD approach is to use sufficiently complex functions (e.g., SVM or deep neural network) to approximate the goal function $F(\theta_1,...,\theta_n)$, while keeping the reports for the agents incentive compatible. The samples are obtained by randomly generating the private parameters $\{\theta_n\}_{n=1,...,N}$. The output of the learning procedure is the functions $h=(g,\mathbf{p})$ and $\{\psi_n\}_n$ in the given forms (e.g., neural network). Traditional studies on AMD \cite{Conitzer2002,Conitzer2004,Guo2010,Sui2013} use heuristic searches. The deep learning approach is employed for AMD in \cite{Dutting2019}, while SVM is applied in \cite{Narasinimhan2016}. The sample complexity of AMD has been analyzed in \cite{Balcan2005}. However, such a learning methodology faces the following severe challenges in the context of spectrum markets for communications and sensing: (a) {Prior Distribution}: In the spectrum market scenario, the major private parameters are the utility functions of the agents. If the utility functions are nonparametric, the sample space is the function space. It is challenging to devise a good prdistribution for the utility functions, instead of for scalar or vector private parameters as in existing mechanism learning. Parameterized functions with predetermined forms may not well generalize. (b) {Offline Learning}: Existing mechanism learning algorithms are mostly offline, without considering the feedbacks of the agents during the operation. It could be more effective to learn the mechanism in an online manner, similarly to reinforcement learning. (c) {Blackbox}: Most existing mechanism learning algorithms are designed in a blackbox manner, without exploiting the intrinsic structure of the mechanism, which substantially decreases the efficiency of mechanism learning.  
Law enactment is a real world practice of mechanism design. A law is seldom set with offline and blackbox computations with artificial distributions. It is of ultimate importance that physical laws of mechanism design be incorporated in the learning process.

\section{System Model}
In this section, we introduce the model of system dynamics and control scheme, based on which a mechanism design framework is formulated. 

\subsection{Model of Dynamics}
We consider a system consisting of $N$ autonomous sub-systems coordinated by a center. We assume a discrete time dynamics, where the state of subsystem $n$ at time $t$ is denoted by $\mathbf{x}_n(t)\in \mathbb{R}^d$. The overall system state is the stack of the individual states, namely $\mathbf{x}=(\mathbf{x}_1^T,...,\mathbf{x}_N^T)^T$. In the generic case, the dynamics of subsystem $n$ is given by
\begin{eqnarray}\label{eq:generic_dynamics}
\mathbf{x}_n(t+1)=f(\mathbf{x}_n(t),\mathbf{u}_n(t),\mathbf{w}_n(t)),
\end{eqnarray}
where $\mathbf{u}_n\in \mathbb{R}^d$ is the control action of subsystem $n$, $\mathbf{w}_n$ is random perturbations and $f$ is the subsystem state evolution law. Note that we assume that the dynamics of subsystem $n$ is not directly impacted by the behaviors of other subsystems; instead, they are coupled indirectly via the coordination of the center. Our future research will extend to the generic case with coupled dynamics of the subsystems. A special but very useful case is the linear dynamics, which is given by
\begin{eqnarray}\label{eq:linear_dynamics}
\mathbf{x}_n(t+1)=\mathbf{A}_{n}\mathbf{x}_n(t)+\mathbf{B}_{n}\mathbf{u}_n(t)+\mathbf{w}_n(t),
\end{eqnarray}

We assume that each subsystem $n$ has a von-Neumann-Morgenstern utility function $U_n$ as a function of the system state and the control actions\footnote{Here we assume that each subsystem is rational and satisfies the four von-Neumann-Morgenstern axioms \cite{Peterson2017}}. A special case of the utility function is the negative of a quadratic function, namely
\begin{eqnarray}\label{eq:LQG0}
U_n(t+1)&=&-\left(\mathbf{x}_n(t+1)-\mathbf{x}_n^0\right)^T\mathbf{Q}_n\left(\mathbf{x}_n(t+1)-\mathbf{x}_n^0\right)\nonumber\\
&-&\mathbf{u}_n(t)^T\mathbf{R}_n\mathbf{u}_n(t),
\end{eqnarray}
where $\mathbf{x}_n^0$, $\mathbf{Q}_n$ and $\mathbf{R}_n$ are the parameters of the utility function. We consider $\mathbf{x}_n^0$ as the desired system state, while the matrices $\mathbf{Q}_n$ and $\mathbf{R}_n$ are assumed to be positive definite. We assume that the coordinator knows $\mathbf{x}_n^0$, but not $\mathbf{Q}_n$ and $\mathbf{R}_n$.

We consider a coordinator which can coordinate the operations of the autonomous subsystems. It is assumed that the coordinator can observe all the actions $\left\{\mathbf{u}_n\right\}_n$ and states of the systems $\left\{\mathbf{x}_n\right\}_n$, while not knowing their utility functions. The coordinator does not control the autonomous systems using direct instructions. Instead, it sets a game for each subsystem $n$ with the payoff $R_n(\mathbf{u}_n(t),\mathbf{x}_n)$, which is given by
\begin{eqnarray}
R_n(\mathbf{u}_n(t),\mathbf{x}_n)=U_n(\mathbf{x},\mathbf{u})+T(\mathbf{x}),
\end{eqnarray}
where $U_n$ is the utility of the subsystem itself and $T$ is an extra reward given to the subsystem. Then, the subsystem $n$ will take an action $\mathbf{u}_n^*$ that maximizes its own reward, which satisfies 
\begin{eqnarray}
\nabla_{\mathbf{u}} R_n(\mathbf{u})|_{\mathbf{u}_n=\mathbf{u}^*_n}=0.
\end{eqnarray}

The mechanism design is how to design the game payoff rule $R(\mathbf{x},\cdot)$, in order to maximize the social welfare, namely
\begin{eqnarray}\label{eq:social_welfare}
\max_{\left\{P_n(\mathbf{x}(t),\cdot)\right\}_n}\sum_{n=1}^N U_n(\mathbf{x}_{n}(t+1),\mathbf{u}_n(t))-\Psi(\mathbf{x}(t+1)),
\end{eqnarray}
where $\Psi$ is a regulation function determined by the overall system state.
Note that such an optimization is carried out for each possible $\mathbf{x}$, or for each time. For simplicity, we consider only this myopic strategy and leave the long-term reward to our future study.

\subsection{Mechanism Design Framework}
\begin{figure}
  \centering
  \includegraphics[scale=0.45]{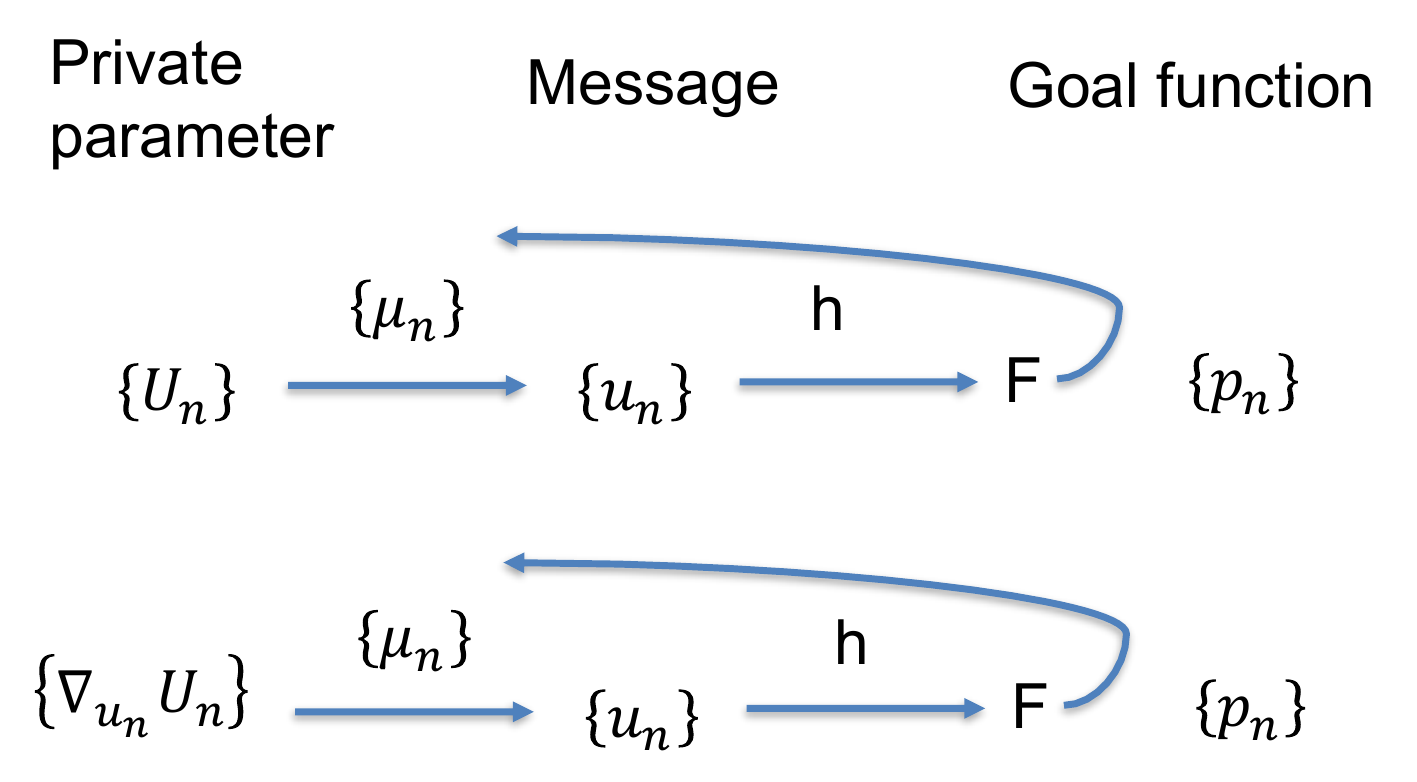}
  \caption{Elements of the mechanism learning.}\label{fig:structure}
  \vspace{-0.1in}
\end{figure}

Using the terminology of mechanism design, the elements of the mechanism are given as follows:
\begin{itemize}
\item Parameter space $\Theta$: The private parameter of subsystem $n$ is the the utility function $U_n$, which is kept to only itself.
\item Game function $g_n$: Since the subsystems are assumed to be rational, the actions are calculated by each subsystem by maximizing the reward function, namely
\begin{eqnarray}\label{eq:game3}
g_n(U_n,\mathbf{u}_n)=\nabla_{\mathbf{u}_n}U_n(\mathbf{x}_n(t+1),\mathbf{u}_n(t))-\mathbf{p}_n(\mathbf{x})
\end{eqnarray}
\item Incentive Compatibility: 
The control action is to maximize the reward and is thus given by
\begin{eqnarray}\label{eq:incentive}
\mathbf{u}_n^*=(\nabla_{\mathbf{u}_n}U_n)^{-1}\left(\mathbf{p}_n^*\right).
\end{eqnarray}
\item Message space $M$: There is no explicit message exchange since the subsystems do not report explicit information about their private parameters. However, the actions taken by the subsystems can be considered as the messages, since they carry information about their own utility functions and are observed by the coordinator. We observe that the output function $F$ actually depends on only the gradients $\{\nabla_{\mathbf{u}_n}U_n\}_n$, thus can also be considered as a vector field of $\mathbb{R}^d$, which facilitates the subsequent geometric argument. Therefore, the necessary information for computing $F$ is the vector fields $\{\nabla_{\mathbf{u}_n}U_n\}_n$ of different subsystems. 
\item Output functional $F$: The output function is the optimal actions, namely
\begin{eqnarray}
F_{\mathbf{x}}\left(\left\{U_n\right\}_n\right)=\{\mathbf{u}_n(\mathbf{x})\}_{n},
\end{eqnarray}
where $\mathbf{x}$ is considered as the parameter while the argument is the set of utility functions. Due to the social welfare function is given in (\ref{eq:social_welfare}), the optimal price is given by
\begin{eqnarray}\label{eq:goal}
\{\mathbf{u}_n(\mathbf{x})\}_{n}^*&=&\arg\max_{\left\{\mathbf{p}_n\right\}_n}\sum_{n=1}^N U_n(\mathbf{x}_{n}(t+1),\mathbf{u}_n)\nonumber\\
&-&\Psi(\mathbf{x}(t+1)),
\end{eqnarray}
where $\mathbf{x}_n(t+1)$ is a function of $\mathbf{x}_n(t)$ and $\mathbf{u}_n(\mathbf{p}_n)$, thus being a function of $\mathbf{u}_n(\mathbf{p})$ is determined by (\ref{eq:incentive}).

\end{itemize}
These basic elements in the mechanism design are illustrated in Fig. \ref{fig:structure}.
The major challenges to the mechanism design in the context of coordinated control are the following:
\begin{itemize}
\item Complex Expression: The explicit expression for the optimal control action, which is determined by (\ref{eq:goal}) and (\ref{eq:incentive}), could be very complicated, or even prohibitive.

\item Information Deficiency: The goal function is based on the knowledge of utility functions. However, the systems may not disclose them directly to the coordinator. It is challenging to uncover the information from the observed control actions.
\end{itemize}

\subsection{Fiber Bundle Modeling}

Different from the standard mechanism problems, the mechanism of the coordinated control is dynamic, namely the vector field $\nabla_{\mathbf{u}_n} U_n$ changes with time, since $U_n$ is dependent on the position $\mathbf{x}_n$. Therefore, we can consider each point $(\mathbf{x}_n,\mathbf{u}_n)$ in $\mathbf{R}^{2d}$ carrying a vector field $\nabla_{\mathbf{u}_n} U_n(\mathbf{x}_n)$. This forms a fiber bundle\footnote{Informally speaking, a fiber bundle means attaching to each point $\mathbf{x}$ in the base manifold a structure called fiber.}, where the fiber at each point $\mathbf{x}_n$ in $\mathbb{R}^d$ is a vector field $\nabla_{\mathbf{u}_n} U_n(\mathbf{x}_n)$ over the space of $\mathbf{u}_n$. The generic case is  illustrated in Fig. \ref{fig:bundle} (a), while the case of $x_n\in \mathbb{R}$ and $\mathbf{u}_n\in \mathbb{R}^2$ is shown in Fig. \ref{fig:bundle} (b). Such a geometric model will be used throughout this paper.

\begin{figure}
  \centering
  \includegraphics[scale=0.4]{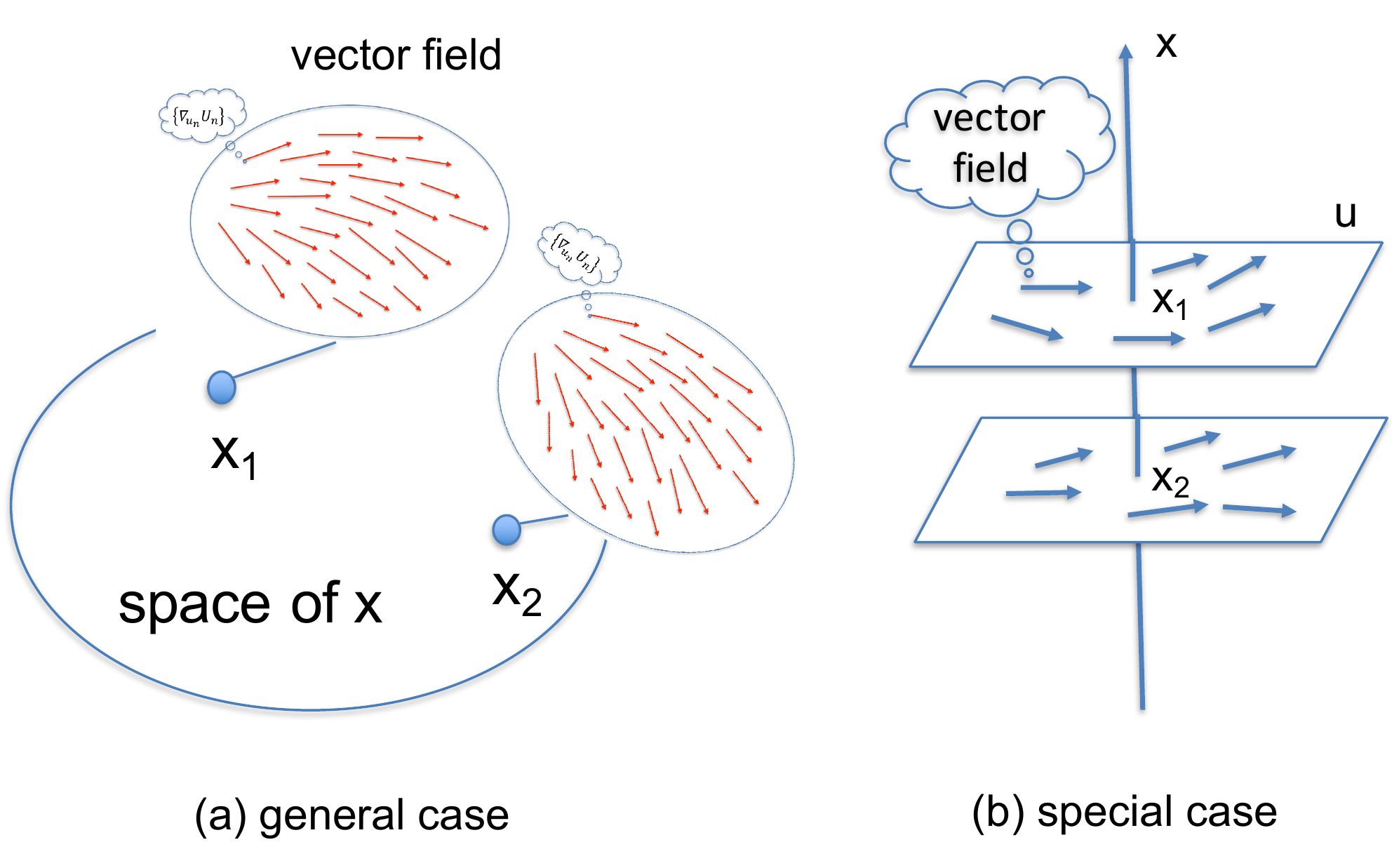}
  \caption{Illustration of the bundle of vector fields.}\label{fig:bundle}
  \vspace{-0.1in}
\end{figure}

\section{Mechanism Learning: Stationary Case}\label{sec:stationary}
In this section, we consider the case in which the learning procedure of mechanism is much faster than the physical dynamics, such that we can assume that the system state $\mathbf{x}$ is constant. This assumption is valid for cases of slow dynamics. For theoretical analysis, we assume that, given the games set by the coordinator, the subsystems reveal their actions but do not take real actions. Only when the `virtual' actions converge to a stationary one, real actions are taken to proceed to the next stage. This polling-action procedure is illustrated in Fig. \ref{fig:polling} and can also allow real actions during the polling procedure. 

\begin{figure}
  \centering
  \includegraphics[scale=0.45]{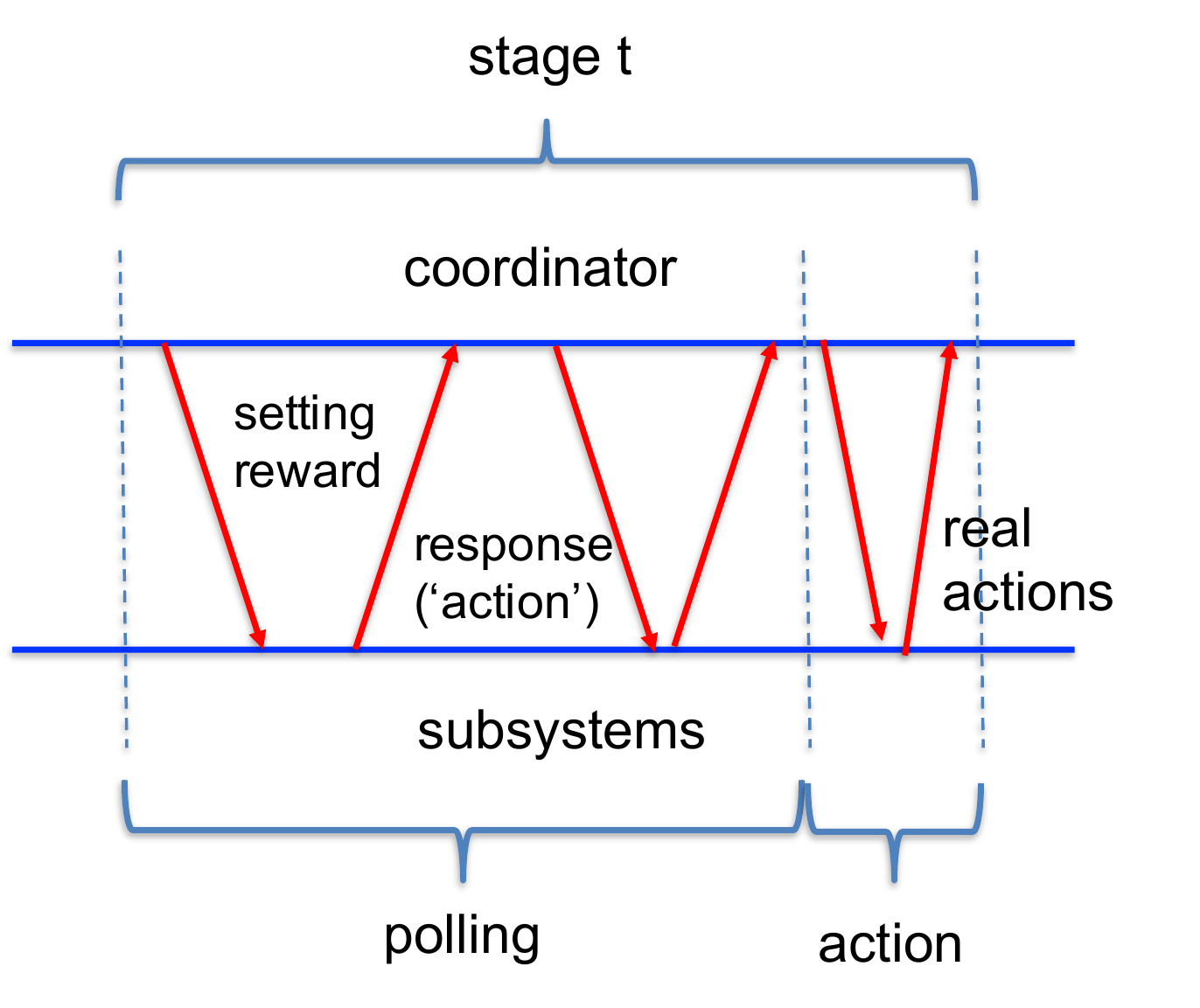}
  \caption{The polling-action mechanism in stationary situations.}\label{fig:polling}
  \vspace{-0.1in}
\end{figure}

\subsection{Strategies of learning}
The key elements in the mechanism learning include the message functions $\{\mu_n\}$ and the decision function $h$. The received message is the observation on the control action $\mathbf{u}_n$, from which the coordinator needs to extract the essential information needed for the computation, namely the gradients $\{\nabla_{\mathbf{u}_n}U_n\}$ for the decision making $h$. Once $\{\nabla_{\mathbf{u}_n}U_n\}$ is extracted from the messages, the computation of $h$, namely the output prices, is carried out by using (\ref{eq:goal}), which is straightforward. Therefore, the key challenge in the mechanism learning is to extract the information on the gradients $\{\nabla_{\mathbf{u}_n}U_n\}$ from the observed actions $\mathbf{u}_n$. The following three strategies for the learning procedure will be adopted and will be detailed subsequently:
\begin{itemize}
\item Parametric learning: We assume that the utility functions are quadratic functions (thus the vector bundle has a special structure) and then learn the parameters from the observed actions. Note that quadratic utility functions are widely used in control theory and economics. It also serves as a `reference' model for nonparametric cases. 
\item Nonparametric learning based on fictitious play \cite{Fudenberg1998}: We do not need to estimate the utility function; instead, we allow the sub-systems to play the fictitious play and reach the corresponding Nash equilibrium in (\ref{eq:game}). 
\end{itemize}

\subsection{Parametric Learning: Quadratic Utilities}
We begin from the case in which the utility functions are assumed to be quadratic, given in (\ref{eq:LQG0}). Therefore, the private information of subsystem $n$ is the matrices $\mathbf{Q}_n$ and $\mathbf{R}_n$. 

\subsubsection{Message Space}
We first study the dimension of message space for the mechanism design. The following theorem shows that, in order to achieve the goal function, all the parameters (elements in the matrices $\left\{\mathbf{Q}_n\right\}_n$ and $\left\{\mathbf{R}_n\right\}_n$) need to be learned. The proof is given in Appendix \ref{appdx:proof_quadratic}.

\begin{theorem}\label{thm:full_quadratic}
For the case of quadratic utility functions, the minimum dimension of the message space is given by
\begin{eqnarray}
dim(M)=2Nd^2.
\end{eqnarray}
\end{theorem}

\subsubsection{Algorithm}

\subsection{Nonparametric Approach: Coordinated Fictitious Play}
Now, we consider the nonparametric utility functions. Although the utility functions are arbitrary with mild constraints (e.g., concavity and smoothness), it is not necessary to estimate the utility functions (or the corresponding vector fields) perfectly. One approach to avoid the explicit parameter estimation is the fictitious play, namely the subsystems play a game by setting the optimal response to the current actions of opponents. One expects the dynamics of fictitious play will converge to the Nash equilibrium (although not necessary) and thus achieve the optimal mechanism. Note that the subsystems do not interact with each other directly. Therefore, the game will be set by the coordinator adaptively. 

\subsubsection{Full Step Fictitious Play}
The coordinator sets the regulation term $\Psi$ as an extra reward for subsystem $n$, while fixing the system states of other subsystems. Then, at time $t$, subsystem plays a game with the following payoff
\begin{eqnarray}\label{eq:fictitious}
R_n(\mathbf{u}(t))&=&U_n(\mathbf{x}(t+1),\mathbf{u}(t))\nonumber\\
&+&\Psi(\mathbf{x}_n(t+1),\{\mathbf{x}_k(t)\}_{k\neq n}).
\end{eqnarray}

\begin{figure}
  \centering
  \includegraphics[scale=0.45]{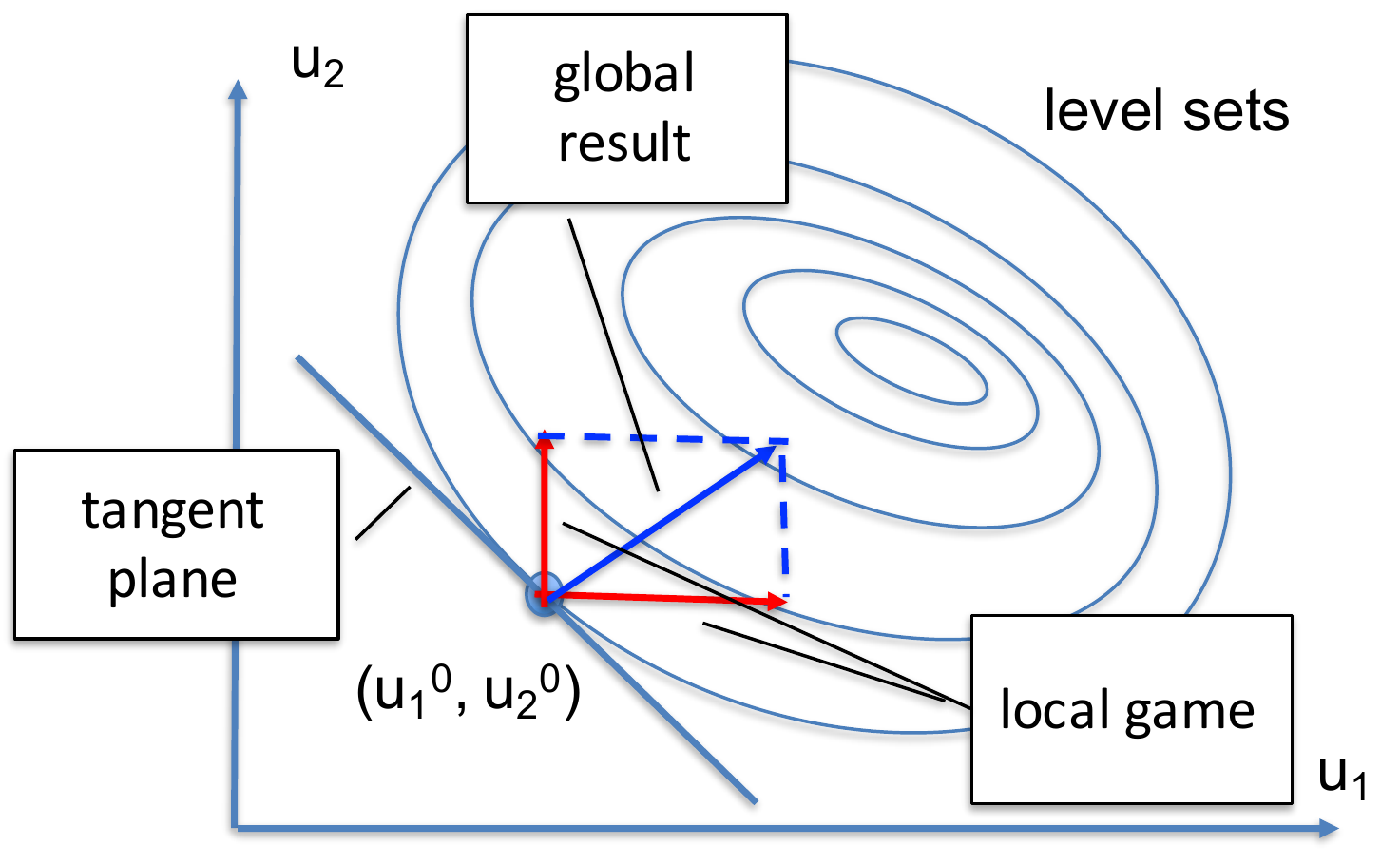}
  \caption{Illustration of the fictitious play and level sets.}\label{fig:level_set}
  \vspace{-0.1in}
\end{figure}

The action of system $n$, $\mathbf{u}_n^{t}$, is given by the solution of the following equation:
\begin{eqnarray}\label{eq:fictitious}
\nabla_{\mathbf{u}_n}U_n+\nabla_{\mathbf{u}_n}R(\cdot|\mathbf{u}_{-n}(t-1))=0,\qquad n=1,...,N.
\end{eqnarray}
If there are multiple solutions to (\ref{eq:fictitious}), we select the one closest to $\mathbf{u}_n^{t-1}$. 

The following theorem describes the convergence of fictitious play to the Nash equilibrium, thus the optimal mechanism. The proof is given in Appendix \ref{appdx:simultaneous}.
\begin{theorem}\label{thm:simultaneous}
Suppose that all Hessian matrices are positive definite. Moreover, the gradient is of order $\|\mathbf{u}\|$. When the derivatives higher than the second order are sufficiently small, and the partial derivatives $\frac{\partial^2 R}{\partial u_{ij}\partial u_{mn}}$ ($m\neq i$) are sufficiently small, the above fictitious play converges to the optimal point. 
\end{theorem}
\begin{remark}
The condition in the theorem means that the coupling of the sub-systems are weak. The cost functions of each subsystem is close to a quadratic function. Therefore, the level sets are more similar to spheres. The local optimizations result in directions parallel to the subspaces of the subsystems. Fig. \ref{fig:level_set} shows that the sum of the individual directions can still improve the objective function value. 
\end{remark}

Instead of the above simultaneous update, the coordinator can also control the games such that the subsystems update their actions in turns. At time $t$, only subsystem $n=mod(t,N)+1$ updates its action:
\begin{eqnarray}\label{eq:fictitious}
\nabla_{\mathbf{u}_n}U_n+\nabla_{\mathbf{u}_n}R(\cdot|\mathbf{u}_{-n}(t-1))=0.
\end{eqnarray}

\subsubsection{Incremental Fictitious Play}
Simulation shows that full step size and simultaneous fictitious play may not converge, which is a common phenomenon of fictitious play []. Therefore, it is desirable to carry out the fictitious play simultaneously. We follow the framework of variational inequality, which is defined as follows: Given a subset $K$ of $\mathbb{R}^n$ and a mapping $F:K\rightarrow \mathbb{R}^n$, find an $\mathbf{x}\in K$ such that
\begin{eqnarray}\label{eq:VI}
(\mathbf{y}-\mathbf{x})^TF(\mathbf{x})\geq 0,\qquad \forall \mathbf{y}\in K.
\end{eqnarray}
If we consider $F$ as the vector field $(-\nabla_{\mathbf{u}_1}(U_1+\Psi),...,-\nabla_{\mathbf{u}_N}(U_N+\Psi))\in \mathbb{R}^{Nd}$, then $\mathbf{x}$ satisfying (\ref{eq:VI}) is the solution to the Nash equilibrium (Section 1.4.2. in \cite{Facchinei2003}). In \cite{Facchinei2003}, the following simple projection iteration is proposed:
\begin{eqnarray}\label{eq:projection}
\mathbf{x}^{k}=\Pi_{K}(\mathbf{x}^k+\tau_kF(\mathbf{x}^k)).
\end{eqnarray}
The iterative projection is shown to converge to the solution when $F$ is monotone, namely
\begin{eqnarray}
\left(F(\mathbf{x})-F(\mathbf{y})\right)^T\left(\mathbf{x}-\mathbf{y}\right)\geq c\|F(\mathbf{x})-F(\mathbf{y}\|^2,
\end{eqnarray}
where $c$ is a positive constant, and the step $\tau_k$ is set 
\begin{eqnarray}
0<\inf_k \tau_k\leq \sup_k\tau_k<2c.
\end{eqnarray}

Following the framework of variational inequality, we consider 
\begin{eqnarray}
F=(\nabla_{\mathbf{u}_1}R_1(\mathbf{u}),...,\nabla_{\mathbf{u}_N}R_N(\mathbf{u})),
\end{eqnarray}
where $R_n(\mathbf{u})=U_n(\mathbf{u}_n)+\Psi(\mathbf{u})$.

However, the incremental update, similar to the gradient descent approach, cannot be directly applied to the fictitious play, since incentives are needed to make the subsystems update their actions incrementally due to the myopic assumption. Moreover, in the variational inequality framework, the vector field $F$ is known, while the corresponding vector field is only partially known in our context (the gradient of $\Psi$ is known, while that of $U_n$ is unknown). This can be implemented in the design of games. The following two approaches are proposed in this paper:
\begin{itemize}
\item Two-stage play: In the $k$-th round, two stages of games are carried out. In the first stage, the game reward for each subsystem $n$ is set to
\begin{eqnarray}\label{eq:reward1}
R_n(\mathbf{u}_n)&=&U_n(\mathbf{u}_n)+\Psi(\mathbf{u}_n,\mathbf{u}_{-n}^{k-1})\nonumber\\
&-&\lambda_k\|\mathbf{u}_n-\mathbf{u}_n^{k-1}\|^2.
\end{eqnarray}
Each subsystem $n$ takes action $\hat{\mathbf{u}}_n$ by solving
\begin{eqnarray}\label{eq:stage1}
&&\nabla_{\mathbf{u}_n} U_n(\mathbf{u}_n^k)+\nabla_{\mathbf{u}_n}\Psi(\mathbf{u}_n^k,\mathbf{u}_{-n}^{k-1})\nonumber\\
&-&\lambda_k(\mathbf{u}_n^k-\mathbf{u}_n^{k-1})=0.
\end{eqnarray}
Then, in the same stage, the reward is set to
\begin{eqnarray}
R_n(\mathbf{u}_n)&=&\|\mathbf{u}_n\|^2-\mathbf{u}_n^T\mathbf{u}_n^{k-1}\nonumber\\
&-&2\gamma_k\mathbf{u}_n^T\nabla_{\mathbf{u}_n}\Psi(\hat{\mathbf{u}}_n^k,\mathbf{u}_{-n}^{k-1})\nonumber\\
&+&2\gamma_k\lambda_k\mathbf{u}_n^T(\mathbf{u}_n-\mathbf{u}_n^{k-1})\nonumber\\
&+&2\gamma_k\mathbf{u}_n^T\nabla_{\mathbf{u}_n}\Psi(\hat{\mathbf{u}}^k),
\end{eqnarray}
which results in 
\begin{eqnarray}
\mathbf{u}_n^k&=&\hat{\mathbf{u}}_n^{k-1}-\gamma_k\nabla_{\mathbf{u}_n}\Psi(\hat{\mathbf{u}}_n^k,\mathbf{u}_{-n}^{k-1})\nonumber\\
&+&\lambda_k(\mathbf{u}_n-\mathbf{u}_n^{k-1})+\gamma_k\nabla_{\mathbf{u}_n}\Psi(\hat{\mathbf{u}}^k)\nonumber\\
&=&\hat{\mathbf{u}}_n^{k-1}+\gamma_k\nabla_{\mathbf{u}_n}U_n(\hat{\mathbf{u}}_n^k)+\gamma_k\nabla_{\mathbf{u}_n}\Psi(\hat{\mathbf{u}}^k),
\end{eqnarray}
where the last equality is due to (\ref{eq:stage1}).

\item Single-stage play: Before the operation, initialize $\tilde{\mathbf{u}^0}$. In the $k$-th round, we set
\begin{eqnarray}\label{eq:probe}
R_n(\mathbf{u}_n)=U_n(\mathbf{u}_n)+\Psi(\mathbf{u}_n,\tilde{\mathbf{u}}_{-n}^{k-1})-\lambda_k\|\mathbf{u}_n-\mathbf{u}_n^{k-1}\|^2.
\end{eqnarray}
Given the responses $\{\mathbf{u}_n^k\}_{n=1,...,N}$, then the coordinator calculates
\begin{eqnarray}
\tilde{\mathbf{u}}^{k}&=&{\mathbf{u}}^{k-1}-\gamma_k \nabla_{\mathbf{u}_n}\Psi(\mathbf{u}_n^k,\tilde{\mathbf{u}}_{-n}^{k-1})\nonumber\\
&+&\lambda_k(\mathbf{u}_n-\mathbf{u}_n^{k-1})+\gamma_k\nabla_{\mathbf{u}_n}\Psi(\mathbf{u}^{k})\nonumber\\
&=&{\mathbf{u}}^{k-1}+\gamma_k \nabla_{\mathbf{u}_n}U_n(\mathbf{u}_n^{k})+\gamma_k\nabla_{\mathbf{u}_n}\Psi(\mathbf{u}^{k}).
\end{eqnarray}
where the last step is due to (\ref{eq:probe}).
\end{itemize}
\begin{remark}
Some rationales in the algorithms are given as follows:
\begin{itemize}
\item The term $\lambda_k\|\mathbf{u}_n-\mathbf{u}_n^{k-1}\|^2$ in (\ref{eq:reward1}) and (\ref{eq:probe}) is to assure that the new point $\hat{\mathbf{u}}^k$ and ${\mathbf{u}}^k$ not too far away from ${\mathbf{u}}^{k-1}$ and $\tilde{\mathbf{u}}^{k-1}$, respectively. This is similar to the Tikhonov regularization and the proximal point algorithms.

\item The action optimizations in (\ref{eq:reward1}) and (\ref{eq:probe}) are to evaluate the vector field at a given $\mathbf{u}$, by leveraging the known expression of $\Psi$, thus facilitating the variational inequality approach in (\ref{eq:projection}).
\end{itemize}
\end{remark}

For the two-stage approach, we define $\mathbf{e}_n(\mathbf{u})$ as
\begin{eqnarray}
\mathbf{e}_n(\mathbf{u}_n)=\nabla_{\mathbf{u}_n} R_n (\mathbf{\hat{u}})-\nabla_{\mathbf{u}_n} R_n (\mathbf{{u}}),
\end{eqnarray}
where $\hat{\mathbf{u}}_n$ is obtained from $\mathbf{u}_n$ according to (\ref{eq:stage1}), namely
 \begin{eqnarray}\label{eq:stage1}
&&\nabla_{\mathbf{u}_n} U_n(\hat{\mathbf{u}}_n)+\nabla_{\mathbf{u}_n}\Psi(\hat{\mathbf{u}}_n,\mathbf{u}_{-n})\nonumber\\
&-&\lambda_k(\hat{\mathbf{u}}_n-\mathbf{u}_n)=0.
\end{eqnarray}

And we need the following definition.
\begin{definition}
We say that a function $F:\mathbb{R}^n\rightarrow \mathbb{R}^n$ is co-coercive with constant $c>0$, if
\begin{eqnarray}
(F(\mathbf{x})-F(\mathbf{y}))^T(\mathbf{x}-\mathbf{y})\geq c\|F(\mathbf{x})-F(\mathbf{y})\|_2^2.
\end{eqnarray}
\end{definition}

\begin{theorem}
For the two-stage approach, suppose that both the functions $F$ and $\mathbf{e}$ are co-coercive with constants $c_1$ and $c_2$, respectively. If 
\begin{eqnarray}
0<\inf_k \tau_k<\sup_k\tau_k<2\min(c_1,c_2),
\end{eqnarray}
the two-stage algorithm converges to the Nash equilibrium.
\end{theorem}

We also notice that we need to select a large $\lambda_k$, such that $\mathbf{e}(\mathbf{u})$ is sufficiently small. The following lemma shows a bound on the gap of $\hat{\mathbf{u}}$ and $\mathbf{u}$ as a function of $\lambda_k$.
\begin{lemma}
For the solution to Eq. (\ref{eq:stage1}), we have
\begin{eqnarray}
\|\hat{\mathbf{u}}^{k}-\mathbf{u}^{k-1}\|\leq \frac{ND}{\lambda_k},
\end{eqnarray}
where
\begin{eqnarray}
D=\max_{n}\sup_{\mathbf{u}_n}(\nabla_{\mathbf{u}_n}U_n(\mathbf{u}_n)+\nabla_{\mathbf{u}_n}\Phi(\mathbf{u})).
\end{eqnarray}
\end{lemma}
\begin{proof}
The proof is straightforward. From (\ref{eq:stage1}), we have 
\begin{eqnarray}
\lambda_k(\mathbf{u}_n^k-\mathbf{u}_n^{k-1})=\nabla_{\mathbf{u}_n} U_n(\mathbf{u}_n^k)+\nabla_{\mathbf{u}_n}\Psi(\mathbf{u}_n^k,\mathbf{u}_{-n}^{k-1}),
\end{eqnarray}
which results in 
\begin{eqnarray}
\|\mathbf{u}_n-\mathbf{u}_n^{k-1}\|&=&\frac{1}{\lambda_k}\left\|\nabla_{\mathbf{u}_n} U_n(\mathbf{u}_n^k)+\nabla_{\mathbf{u}_n}\Psi(\mathbf{u}_n^k,\mathbf{u}_{-n}^{k-1})\right\|\nonumber\\
&\leq &\frac{D}{\lambda_k}.
\end{eqnarray}
This concludes the proof by stacking the inequalities for $n=1,...,N$.
\end{proof}

A simpler approach for the incremental fictitious play is to employ the Tikhonov regularization term directly:
\begin{eqnarray}\label{eq:probe2}
\mathbf{u}_{n}^k&=&\arg\max_{\mathbf{u}}\left(U_n(\mathbf{u}_n)+\Psi(\mathbf{u}_n,\tilde{\mathbf{u}}_{-n}^{k-1})\right.\nonumber\\
&-&\left.\lambda_k\|\mathbf{u}_n-\mathbf{u}_n^{k-1}\|^2\right).
\end{eqnarray}

This is related to the projection based variational inequality approach. The optimal solution $\mathbf{u}_n^{k}$ satisfies the following equation:
\begin{eqnarray}
\lambda_k(\mathbf{u}_n-\mathbf{u}_n^{k-1})+\nabla R(\mathbf{u}_n,\mathbf{u}_{n-}^{k-1})=0.
\end{eqnarray}
When $\lambda_k$ is sufficiently large, the solution $\mathbf{x}_n^k$ should be sufficiently close to $\mathbf{x}_n^{k-1}$. Therefore, we can expand $\nabla R$ as
\begin{eqnarray}
\nabla R(\mathbf{u}_n,\mathbf{u}_{n-}^{k-1})&=&\nabla R(\mathbf{u}_n^{k-1},\mathbf{u}_{n-}^{k-1})+H_R(\mathbf{u}_n-\mathbf{u}_n^k)\nonumber\\&+&o(\|\mathbf{u}_n-\mathbf{u}_n^k\|),
\end{eqnarray}
where $H_R$ is the Hessian matrix of $R$. Therefore, the update is 
\begin{eqnarray}
\mathbf{u}_n&=&\mathbf{u}_n^{k-1}+\frac{1}{\lambda_k}\nabla R(\mathbf{u}_n^{k-1},\mathbf{u}_{n-}^{k-1})+\frac{1}{\lambda_k}H_R(\mathbf{u}_n-\mathbf{u}_n^k)\nonumber\\
&+&o(\|\mathbf{u}_n-\mathbf{u}_n^k\|)\nonumber\\
&=&\mathbf{u}_n^{k-1}+\frac{1}{\lambda_k}\nabla R(\mathbf{u}_n^{k-1},\mathbf{u}_{n-}^{k-1})+O\left(\frac{1}{\lambda_k}\|\mathbf{u}_n-\mathbf{u}_n^k\|)\right),
\end{eqnarray}
which is a perturbed version of the projection based variational inequality approach.

\section{Dynamic Mechanism Learning}
In this section, we consider the case in which the system states of subsystems change quickly; therefore, there is no time for the polling procedure in Section \ref{sec:stationary}. The main challenge is how to leverage the samples obtained in previous operation, although the system state has been changed; otherwise, there is no sample to exploit for the current decision. Hence, it is of key importance to find connections between the current environment and the history. In this section, we consider two approaches for this connection, namely the unchanged parameters and geometric connection in the vector fields.

\subsection{Parametric Approach: Quadratic Utilities}
Similarly to Section \ref{sec:stationary}, we assume that the utility functions are quadratic. The constant parameters play the role of connection between the current environment and the history, thus providing samples from the history for the current decision.

For computing the gradient, we need to use the fact $\nabla_\mathbf{x}\mathbf{x}^T\mathbf{A}\mathbf{x}=(\mathbf{A}+\mathbf{A}^T)\mathbf{x}$ and $\nabla_\mathbf{x}(f\circ g)(\mathbf{x})=(Dg(\mathbf{c}))^T(\nabla_{\mathbf{y}} f(\mathbf{y}))$, where $D$ is the Jacobian matrix.
Then, for agent $n$, we have
\begin{eqnarray}
\nabla_{\mathbf{u}_n}U_n(\mathbf{x}_n(t+1),\mathbf{u}_n)&=&-2\mathbf{B}_n^T\mathbf{Q}_n(\mathbf{x}_n(t+1)-\mathbf{x}_n^0)\nonumber\\
&-&2\mathbf{R}_n\mathbf{u}_n(t),
\end{eqnarray}
which implies that the optimal action $\mathbf{u}_n$ satisfies
\begin{eqnarray}
\mathbf{p}_n&=&2\mathbf{B}_n^T\mathbf{Q}_n\mathbf{A}_{nn}\mathbf{x}_n(t)+2\left(\mathbf{B}_n^T\mathbf{Q}_n\mathbf{B}_n+\mathbf{R}_n\right)\mathbf{u}_n(t)\nonumber\\
&-&2\mathbf{B}_n^T\mathbf{Q}_n\mathbf{x}_n^0
\end{eqnarray}
Therefore, we have
\begin{eqnarray}\label{eq:CD}
\mathbf{C}_n\mathbf{x}_n(t)+\mathbf{D}_n\mathbf{u}_n(t)=\mathbf{y}_n(t).
\end{eqnarray}
where 
\begin{eqnarray}
\left\{
\begin{array}{lll}
&\mathbf{C}_n=2\mathbf{B}_n^T\mathbf{Q}_n\mathbf{A}_{nn}\\
&\mathbf{D}_n=2\left(\mathbf{B}_n^T\mathbf{Q}_n\mathbf{B}_n+\mathbf{R}_n\right)\\
&\mathbf{y}_n=\mathbf{p}_n+2\mathbf{B}_n^T\mathbf{Q}_n\mathbf{x}_n^0
\end{array}
\right..
\end{eqnarray}
We can rewrite (\ref{eq:CD}) as
\begin{eqnarray}\label{eq:CD0}
\mathbf{E}_n\mathbf{z}_n(t)=\mathbf{f}_n(t),
\end{eqnarray}
where $\mathbf{E}_n=(\mathbf{C}_n,\mathbf{D}_n)$ and $\mathbf{z}_n(t)=(\mathbf{x}_n^T(t),\mathbf{u}_n^T(t))^T$. Stacking the observations in $M$ rounds to one equation, we have
\begin{eqnarray}\label{eq:CD1}
\mathbf{E}_n\mathbf{Z}_n=\mathbf{F}_n,
\end{eqnarray}
where $\mathbf{Z}_n=(\mathbf{z}_n(1),...,\mathbf{z}_n(M))$ and $\mathbf{F}_n=(\mathbf{f}_n(1),...,\mathbf{f}_n(M))$. We further rewrite (\ref{eq:CD1}) as
\begin{eqnarray}
(\mathbf{I}\otimes\mathbf{Z}_n)vec(\mathbf{E}_n)=vec(\mathbf{F}_n),
\end{eqnarray}
which results in
\begin{eqnarray}
vec(\mathbf{E}_n)=(\mathbf{I}\otimes\mathbf{Z}_n)^{-1}vec(\mathbf{F}_n).
\end{eqnarray}
Once $\mathbf{E}_n$ is obtained, we obtain $\mathbf{C}_n$ and $\mathbf{D}_n$ immediately. The parameters $\mathbf{Q}_n$ and $\mathbf{R}_n$ are obtained as follows:
\begin{eqnarray}
\left\{
\begin{array}{ll}
&\mathbf{Q}_n=\frac{1}{2}\mathbf{B}_n^{-1}\mathbf{C}_n\\
&\mathbf{R}_n=\mathbf{D}_n-\mathbf{C}_n
\end{array}
\right..
\end{eqnarray}
We also obtain the mapping from the price to the individually optimal control action:
\begin{eqnarray}
\mathbf{u}_n(t)=\mathbf{D}_n^{-1}\mathbf{y}_n(t)-\mathbf{D}_n^{-1}\mathbf{C}_n\mathbf{x}_n(t).
\end{eqnarray}

Once learning the response of subsystems given the prices, we can optimize the prices in order to maximize the social welfare. We first assume that the penalty term $\Psi$ in (\ref{eq:social_welfare}) is zero. Then, in the social welfare in (\ref{eq:social_welfare}), the term related to the control action $\mathbf{u}_n$ is given by
\begin{eqnarray}
W_n&=&-(\mathbf{x}_n(t+1)-\mathbf{x}_0)^T\mathbf{Q}_n\mathbf{A}_{nn}(\mathbf{x}_n(t+1)-\mathbf{x}_0)\nonumber\\
&-&\mathbf{u}_n^T\mathbf{R}_n\mathbf{u}_n+\Psi(\mathbf{x}(t+1))
\end{eqnarray}
Therefore, we have
\begin{eqnarray}
\nabla_{\mathbf{u}_n}W_n&=&-2\mathbf{B}_n^T\mathbf{Q}_n(\mathbf{x}_n(t+1)-\mathbf{x}_n^0)\nonumber\\
&-&2\mathbf{R}_n\mathbf{u}_n(t)+\mathbf{A}_{nn}^2\nabla_{\mathbf{x}_n}\Psi.
\end{eqnarray}
Suppose that $\mathbf{u}^*$ is the solution to the equation $\nabla_{\mathbf{u}_n}W_n=0$. Then the optimal price $\mathbf{p}_n$ should be set to
\begin{eqnarray}
\mathbf{p}_n^*(t)=\mathbf{C}_n\mathbf{x}_n(t)+\mathbf{D}_n\mathbf{u}_n^*(t)-2\mathbf{B}_n^T\mathbf{Q}_n\mathbf{x}_n^0
\end{eqnarray}

\subsection{Nonparametric Approach: Fictitious Play}

\subsection{Nonparametric Approach: Connection based Fictitious Play}

\begin{figure}
  \centering
  \includegraphics[scale=0.55]{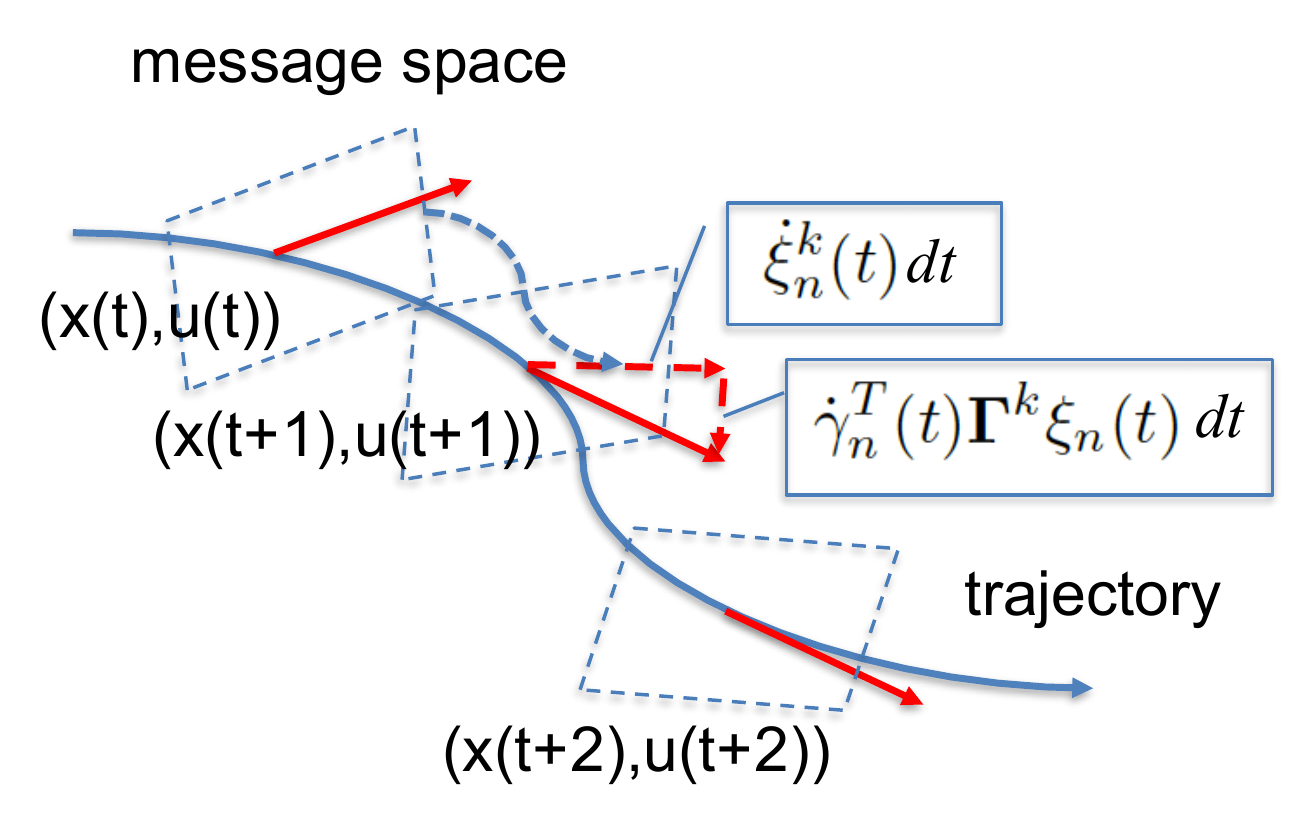}
  \caption{An illustration of trajectory.}\label{fig:trajectory}
  \vspace{-0.1in}
\end{figure}

We denote by $\gamma(t)$ the trajectory of $\mathbf{z}$ as a function of time $t$. Then, when $\mathbf{A}_{nm}$ are sufficiently close to $\mathbf{I}$ and $\mathbf{B}$ is sufficiently small, we can approximate the discrete time system as a continuous time one. We assume that the system of $\xi_n$ is autonomous, which evolves as
\begin{eqnarray}
\nabla_{\dot{\gamma}(t)}\xi_{n}(t)=f(\xi_{n}(t)),
\end{eqnarray}
where $\nabla$ is the covariant derivative. 

To facilitate the tracking of the gradients, we assume affine connection for the covariant derivative. Then, the problem becomes estimating the Christoffel symbol $\Gamma_{ij}^k$ that characterizes the connection via the following relationship:
\begin{eqnarray}
\nabla_{\frac{\partial}{\partial x^i}}\frac{\partial}{\partial x^j}=\sum_{k=1}^M \Gamma_{ij}^k \frac{\partial}{\partial x^k}.
\end{eqnarray}

For the covariant derivative of generic vector fields, we need the following lemma.
\begin{lemma}[Prop.2.2, \cite{Godiinho2014}]
For an affine connection $\nabla$ characterized by the Christoffel symbols $\{\Gamma_{ij}^k\}$, and vector fields $X$ and $Y$ with local coordinates given by
\begin{eqnarray}
\left\{
\begin{array}{ll}
&X=\sum_{i=1}^MX^{i}\frac{\partial }{\partial x^i}\\
&Y=\sum_{i=1}^MY^{i}\frac{\partial }{\partial y^i}
\end{array}
\right.,
\end{eqnarray}
then we have
\begin{eqnarray}\label{eq:connection}
\nabla_X Y=\sum_{k=1}^M\left(X\cdot Y^k+\sum_{i,j=1}^M \Gamma_{ij}^k X^iY^j\right)\frac{\partial}{\partial x^k}
\end{eqnarray}
\end{lemma}

Fixing $k$, we consider $\Gamma^k=(\Gamma_{ij}^k)_{ij}$ as a matrix. Then, we have
\begin{eqnarray}
\frac{D\xi_n^k(t)}{dt}&=&\dot{\xi}_n^k(t)+\sum_{i,j=1}^M \Gamma_{ij}^k \dot{\gamma}_n^i(t)\xi_n^j(t)\nonumber\\
&=&\dot{\xi}_n^k(z(t))+\dot{\gamma}_n^T(t)\mathbf{\Gamma}^k\xi_n(t).
\end{eqnarray}
Here the gradient $\nabla_z\xi_n^k(z(t))$ is dependent on the second order derivative of $U_n$ and is unknown. However, we assume that $U_n$ changes slowly in $\mathbf{z}$ and has a small second order derivative. 

For intuition, we check the covariant derivative in (\ref{eq:connection}). We find that the change of the vector field depends on the following two terms:
\begin{itemize}
\item Linear term $\sum_{k=1}^MX\cdot Y^k\frac{\partial}{\partial x^k}$: the change of the $k$-th component in $\xi_n$ depends on only the derivative of $\xi_n^k$ and the change of the base manifold. 

\item Quadratic term $\sum_{k=1}^M \sum_{i,j=1}^M \Gamma_{ij}^k X^iY^j$: the change of the $k$-th component in $\xi_n$ also depends on the other components in $\xi_n$.
\end{itemize}

\subsubsection{Special Case: Quadratic Utilities}
When the quadratic utility function, as the reference model, is substituted into the expression, we have
\begin{eqnarray}
\xi_{\mathbf{x}}(\mathbf{u})&=&\nabla_{\mathbf{u}_n}U_n(\mathbf{x}_n(t+1),\mathbf{u}_n)\nonumber\\
&=&-2\mathbf{B}_n^T\mathbf{Q}_n(\mathbf{x}_n(t+1)-\mathbf{x}_n^0)-2\mathbf{R}_n\mathbf{u}_n(t),
\end{eqnarray}
Therefore, when $\xi(\cdot,\mathbf{u})$ moves from $\mathbf{x}$ to $\mathbf{x}+\delta \mathbf{x}$, we have
\begin{eqnarray}
\delta \xi(\mathbf{x},\mathbf{u})
=\mathbf{B}_n^T\mathbf{Q}_n\mathbf{A}_n\delta \mathbf{x},
\end{eqnarray}
which is a linear transform of $\delta \mathbf{x}$. Therefore, when the utility functions are quadratic, the affine connection is linear and thus the space is flat.

When there are higher order cross terms in $U(\mathbf{x},\mathbf{u})$, namely
\begin{eqnarray}
\xi_{\mathbf{x}}(\mathbf{u})&=&-2\mathbf{B}_n^T\mathbf{Q}_n(\mathbf{x}_n(t+1)-\mathbf{x}_n^0)-2\mathbf{R}_n\mathbf{u}_n(t)\nonumber\\
&-&\sum_{i}x_{ni} \mathbf{K}_i\mathbf{u}_n.
\end{eqnarray}

Therefore, we have
\begin{itemize}
\item The linear term becomes a constant $\sum_{k=1}^M D_{ki}X_i\frac{\partial}{\partial x^k}$.
\item The cross quadratic term vanishes when there are only quadratic terms in the utility function, and is nonzero when there are higher order cross terms. 
\end{itemize}

When the utility functions are no longer quadratic, the connection can be considered as a combination of linear term $\sum_{k=1}^M D_{ki}X_i\frac{\partial}{\partial x^k}$ and quadratic term, as a refinement on the quadratic utility function $\sum_{k=1}^M \sum_{i,j=1}^M \Gamma_{ij}^k X^iY^j$. Moreover, we assume that the linear coefficients $\left\{D_{ki}\right\}_{ki}$ and the Christoffer symbols $\{\Gamma^k_{ij}\}$ are constants (at least locally). 

\subsubsection{Base Manifold $(\mathbf{x},\mathbf{u})$}
Now, we consider the base manifold as $(\mathbf{x},\mathbf{u})$ and the utility function derivatives as a vector field over $(\mathbf{x},\mathbf{u})$.
Then, we approximate the differential with difference, namely
\begin{eqnarray}
\xi_j(t+1)-\xi_j(t)&\approx&\left(\mathbf{z}(t+1)-\mathbf{z}(t)\right)^T\mathbf{d}_j\nonumber\\
&+&(\mathbf{z}(t+1)-\mathbf{z}(t))^T\Gamma_j\xi(t).
\end{eqnarray}
which is summarized into
\begin{eqnarray}
\delta_{\xi,j}=\Delta_z\mathbf{d}_j+\Delta_z\Gamma_j\xi,
\end{eqnarray}
and
\begin{eqnarray}
\delta_{\xi,j}&=&(\Delta_z\otimes \mathbf{I})\mathbf{d}_j+(\Delta_z\otimes \xi) vec(\Gamma_j)\nonumber\\
&=&(\Delta_z\otimes \mathbf{I},\Delta_z\otimes \xi)(\mathbf{d}_j^T,vec(\Gamma_j)^T)^T,
\end{eqnarray}
which results in
\begin{eqnarray}
(\mathbf{d}_j^T,vec(\Gamma_j)^T)^T=((\Delta_z\otimes \mathbf{I},\Delta_z\otimes \xi))^{+}\delta_{\xi,j}
\end{eqnarray}

\begin{figure}
  \centering
  \includegraphics[scale=0.55]{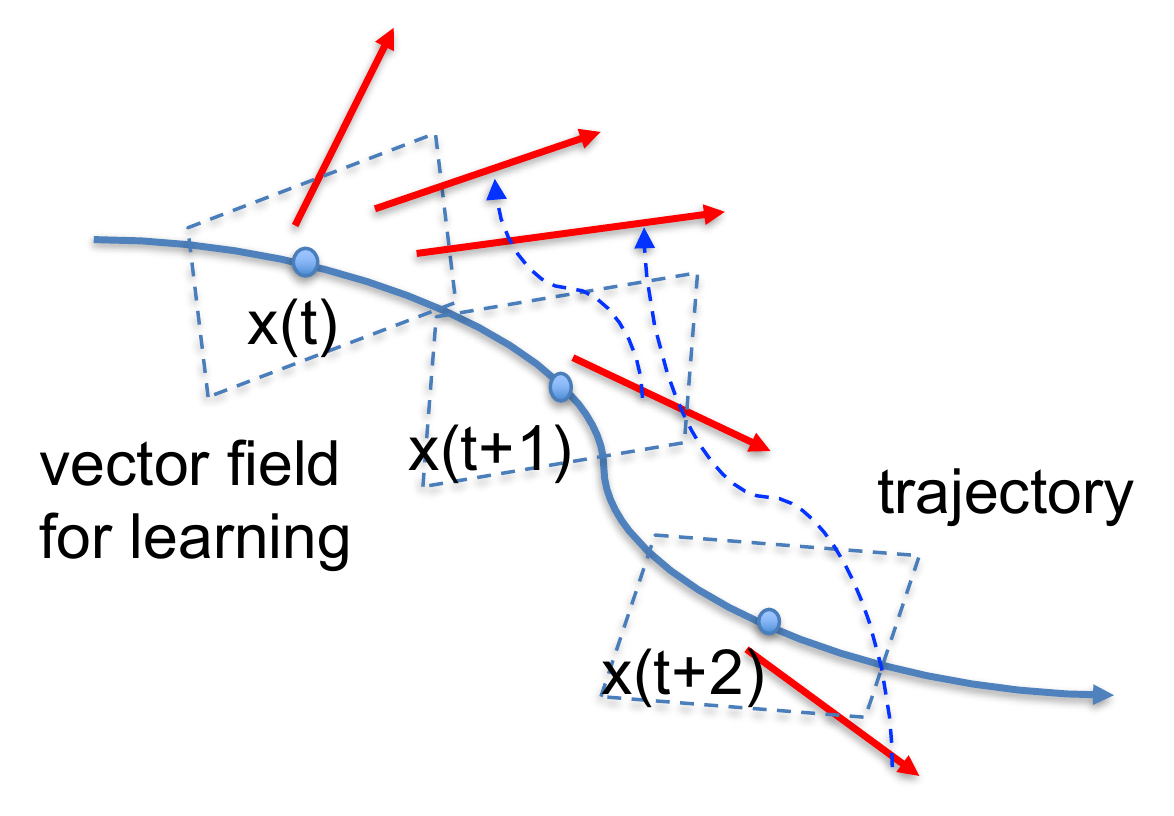}
  \caption{An illustration of pullback and learning.}\label{fig:pullback}
  \vspace{-0.1in}
\end{figure}

\subsection{Special Case: Decomposable Utilities}

\begin{figure}
  \centering
  \includegraphics[scale=0.55]{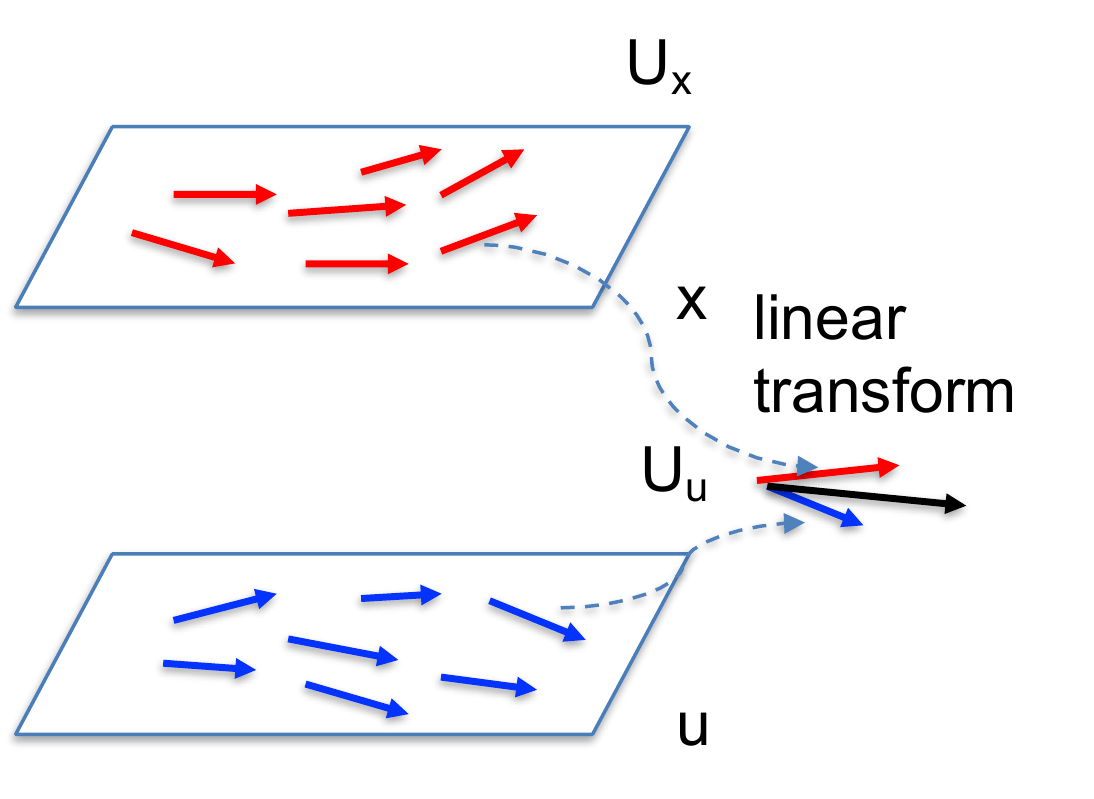}
  \caption{Vector field of decomposable utility functions.}\label{fig:decompose}
  \vspace{-0.1in}
\end{figure}

We assume that the utility function of each agent can be decomposed into two parts:
\begin{eqnarray}
U_n(\mathbf{x}_n,\mathbf{u}_n)=U_n^x(\mathbf{x}_n)+U_n^u(\mathbf{u}_n),
\end{eqnarray}
which is valid for the reference quadratic utility functions. 

The decomposability implies
\begin{eqnarray}\label{eq:dev_Un}
&&\nabla_{\mathbf{u}_n}U_n(\mathbf{x}_n(t+1),\mathbf{u}_n)\nonumber\\
&=&\frac{\partial \mathbf{x}_n(t+1)}{\partial \mathbf{u}_n}\nabla_{\mathbf{x}_m}U_n^x(\mathbf{x}_n)+\nabla_{\mathbf{u}_n}U_n^u(\mathbf{u}_n)\nonumber\\
&=&\mathbf{B}_n\nabla_{\mathbf{x}_n}U_n^x(\mathbf{x}_n)+\nabla_{\mathbf{u}_n}U_n^u(\mathbf{u}_n)
\end{eqnarray}
We observe that the vector field $\nabla_{\mathbf{u}_n}U_n$ over the $2d$-dimensional space $(\mathbf{x}_n,\mathbf{u}_n)$ can be decomposed to two vector fields, namely $\nabla_{\mathbf{x}_n}U_n^x$ over $\mathbb{R}^n$ and $\nabla_{\mathbf{u}_n}U_n^u$ over $\mathbb{R}^n$. Therefore, if we can estimate $\nabla_{\mathbf{x}_n}U_n^x$ and $\nabla_{\mathbf{u}_n}U_n^u$, then we can obtain $\nabla_{\mathbf{u}_n}U_n$ over the much higher base space. However, the challenge is that we cannot observe $\nabla_{\mathbf{x}_n}U_n^x$ and $\nabla_{\mathbf{u}_n}U_n^u$ directly. What we observe is the linear combination of the two vector fields.

Using the vector filed learning discussed in ???, we have
\begin{eqnarray}
\left\{
\begin{array}{ll}
&\hat{\nabla}_{\mathbf{x}_n}U_n^x(\mathbf{x}_n)=\sum_{i=1}^n\mathbf{\Gamma}^x(\mathbf{x},\mathbf{x}_i)\mathbf{c}^x_i\\
&\hat{\nabla}_{\mathbf{u}_n}U_n^u(\mathbf{u}_n)=\sum_{i=1}^n\mathbf{\Gamma}^u(\mathbf{u},\mathbf{u}_i)\mathbf{c}^u_i
\end{array}
\right..
\end{eqnarray}
The goal of learning is to optimize the parameters such that the prediction error is minimized, namely
\begin{eqnarray}
\min_{\mathbf{\Gamma}^x,\mathbf{\Gamma}^u}\left |\sum_{i=1}^n \mathbf{B}_n\hat{\nabla}_{\mathbf{x}_n}U_n^x(\mathbf{x}_n)-\hat{\nabla}_{\mathbf{u}_n}U_n^u(\mathbf{u}_i)-\mathbf{p}\right|^2
\end{eqnarray}

\section{Application in UAM Systems}

\section{Numerical Simulations}\label{sec:numerical}

\section{Conclusions}\label{sec:conclusion}

\appendices
\section{Mechanism Design: A Survey}

\subsection{Generic Mechanism Design}
We assume that there are $N$ agents, each having a local parameter $\theta_n\in \Theta_n$, $n=1,...,N$. A social choice function output $F(\theta_1,...,\theta_n):\prod_{n=1}^N\Theta_n\rightarrow \mathbb{R}^k$ is desired based on the local parameters. Each agent $n$ discloses a message $m_n=\mu_n(\theta_n)$ in a predetermined message space $M$, where $\mu_n$ is its message mapping. Each message is the output of local optimization, namely
\begin{eqnarray}
m_n=\{m|g_n(m,\theta_n)=0\},
\end{eqnarray}
where $g_n(\cdot,\cdot)$ is the function characterizing the optimization taken by agent $n$. Then, the overall message is given by
\begin{eqnarray}
m=\cap_{n=1}^N m_n,
\end{eqnarray}
due to the information decentralization. Then, there exists a function $h:M\rightarrow\mathbb{R}^k$ such that
\begin{eqnarray}
F(\theta_1,...,\theta_n)=h(m).
\end{eqnarray} 
The whole procedure is illustrated in Fig. \ref{fig:illu}.

\begin{figure}
  \centering
  \includegraphics[scale=0.5]{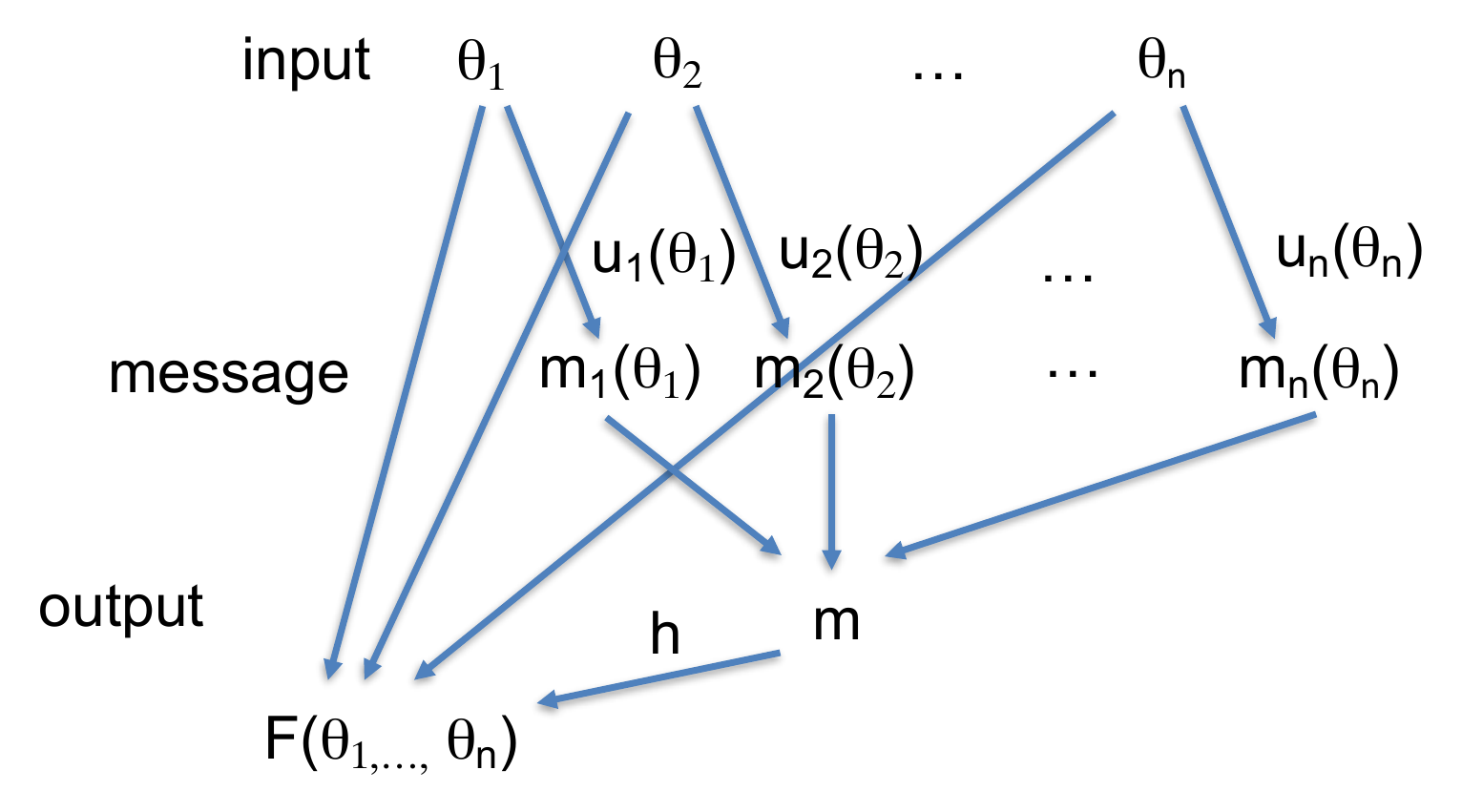}
  \caption{An illustration of economic mechanism.}\label{fig:illu}
  \vspace{-0.1in}
\end{figure}

In the context of economics, the mapping $\mu_i$ is equivalent to a set of equations given by
\begin{eqnarray}
g_n(\theta_n,m)=0,\qquad \theta_n\in \Theta_n, m\in M, n=1,...,N,
\end{eqnarray}
where $g_n:\Theta_n\times M\rightarrow \mathbb{R}^{c_n}$ represent $c_n$ equations. Here $g_n=0$ can be considered as the condition for the optimal solution to an optimization problem (e.g., maximizing local reward).  

The goal of mechanism design is to find the correspondence $\{\mu_i\}_{i=1,...,N}$. An effective approach for mechanism learning is to consider the geometry of economic mechanism.
The goal is to find the maps $\{u_n\}$ and $h$ to realize the goal function $F$, such that $F(\theta_1,...,\theta_N)=h(u_1(\theta_1),...,u_N(\theta_N))$. On one hand, the dimension of the message space may be lower than the parameter space, thus reducing the amount of needed communications; on the other hand, the message function $u_n$, if not 1-to-1 mapping, can prevent agent $n$ from disclosing the full information of $\theta_n$ by reporting $m_n$ instead of $\theta_n$, thus protecting its data privacy. A typical example is the auction for a single good and $N$ auctioneers. The private parameter of auctioneer $n$ is its value on the good $v_n$, while the corresponding message is its bid $b_n$. Note that $b_n$ could be different from $v_n$ since the auctioneer may not want to disclose its valuation on the good. The output function $F$ is the decision on the winner $w$ and the price $p$ that the winner needs to pay.  A good auction mechanism should set the output $(w,p)$ properly, such that the honest report $b_n=v_n$ can maximize the auctioneer $n$'s reward. 

\subsection{Geometric Structure of Mechanism: Single Agent}
It is L. Hurwicz who noticed the geometric structure in the economic mechanism, which was developed by ? Williams to develop the calculus based framework of mechanism design. A quick tutorial is needed to explain this geometric image, by beginning from the single-agent case ($N=1$). Consider a function $F:\mathbb{R}^n\rightarrow\mathbb{R}^m$. It can be represented by its level sets $S(\mathbf{x})=\{\mathbf{x}'|F(\mathbf{x}')=\mathbf{x}\}$, each of which has an index, and is a $d$-dimensional submanifold of $\mathbb{R}^n$ for a certain integer $d$. In the terminology of differential geometry, such sub-manifolds are called foliations. A distribution means a selection of $d$-subspace at the tangent space of each point. When the foliation is 1-dimensional curve ($d=1$), the distribution is simply the tangent lines of the curve. The foliation can be obtained from the distribution by integration. Meanwhile, the distribution is orthogonal to the gradient of the function $F$. The relationships of mapping, foliation and distribution are summarized in Fig. \ref{fig:three}. The diverse representations of the goal function facilitates the design of mechanisms, and play the fundamental role in the proposed research.

\begin{figure}
  \centering
  \includegraphics[scale=0.45]{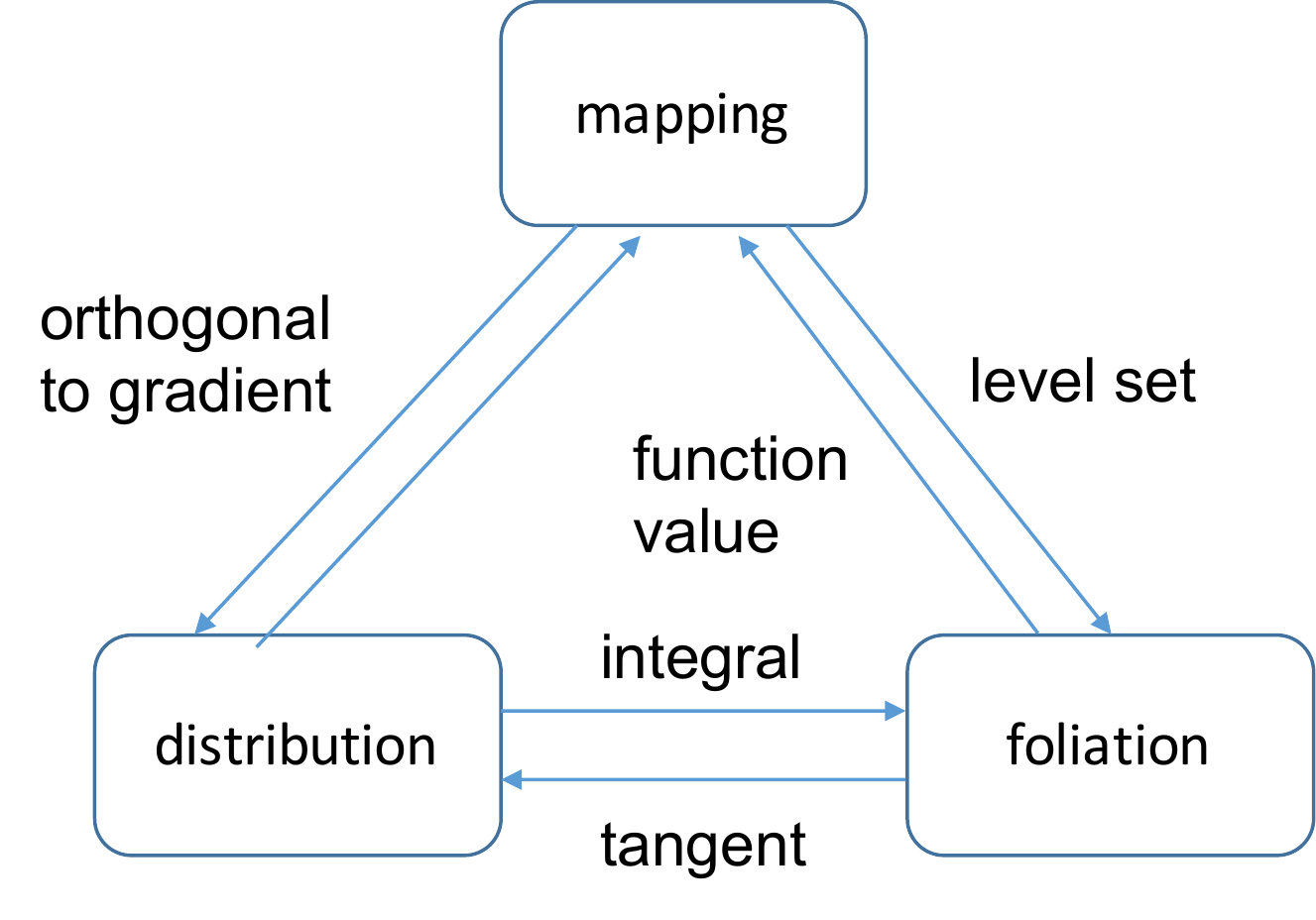}
  \caption{Relationships among mapping, foliation and distribution.}\label{fig:three}
  \vspace{-0.1in}
\end{figure}

A naive example is shown in Fig. \ref{fig:levels} to illustrate the relationships among distribution, foliation and mapping, which is on the allocation of two goods. Here the private parameter is the utility function $U:\mathbb{R}^2_+\rightarrow\mathbb{R}$, whose arguments are the allocations of the two goods. The goal function is $F(U)=U(\mathbf{x})\in \mathbb{R}$, where $\mathbf{x}$ is a fixed point in $\mathbb{R}^2_+$. Then, the agent simply reports the value of utility function, namely $m(U)=U(\mathbf{x})$. In Fig. \ref{fig:levels}, the geometry of the mechanism is shown. The quarter plane $\mathbb{R}_+^2$ is partitioned into level sets, each corresponding to a foliation and resulting the same value of utility (thus being the indifference curve). The utility function value is the index of the level sets. Three level sets of the utility function $U$ are plotted, which are assumed to be 1-dimensional. The distribution $D$ is the tangent lines of the level sets (foliations). The level sets can be obtained by integrating along the tangent lines (the distribution). Note that the optimal price vector (which the gradient of the utility function) is orthogonal to the tangent line, since $\mathbf{p}(\mathbf{x_0})\cdot \mathbf{x}_0\geq \mathbf{p}(\mathbf{x_0})\cdot \mathbf{x}$ (the maximum wealth), and $x_0$ is the tangent point of the indifferent curve and the budget $\mathbf{p}(\mathbf{x_0})\cdot (\mathbf{x}_0-\mathbf{x})=0$.

\begin{figure}
  \centering
  \includegraphics[scale=0.45]{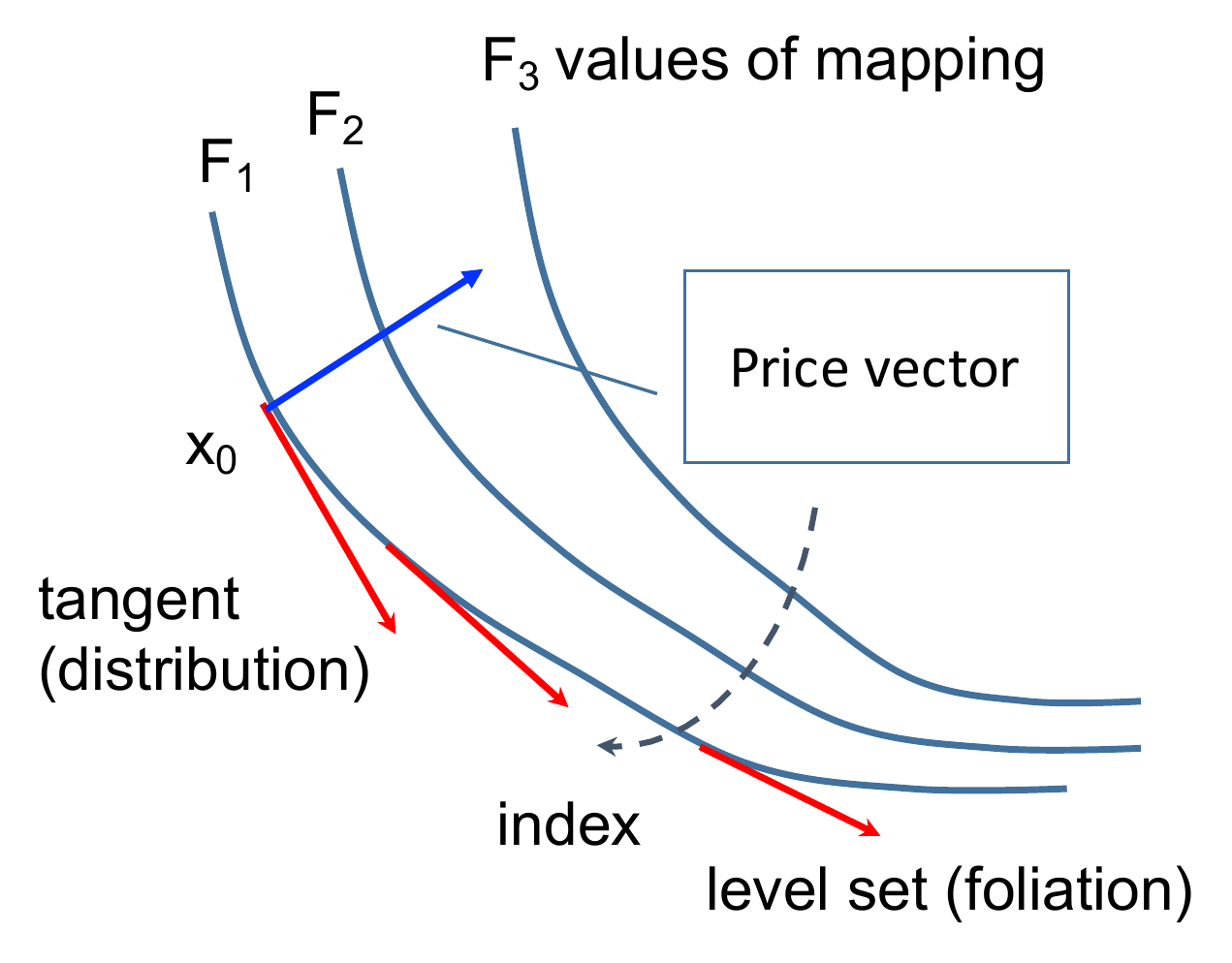}
  \caption{An illustration of distribution, foliation and mapping.}\label{fig:levels}
  \vspace{-0.1in}
\end{figure}

\subsection{Geometric Structure of Mechanism: Multiple Agents}
The above geometric image is for a single agent. In the generic case of multiple agents, a foliation (level set) is replaced with the product of the foliation sets of different agents, while the distribution (tangent) is also replaced with the product of the corresponding distributions. An example is illustrated in Fig. \ref{fig:levels}, where $N=2$, $k_1=2$ and $k_2=1$. Consider the neighborhood of $\mathbf{x}=(\mathbf{x}_1,x_2)\in \mathbb{R}^3$, where $\mathbf{x}_1\in \mathbb{R}^2$ and $x_2\in \mathbb{R}$. Fix $\mathbf{x}_1$, the foliation $S_2$ for agent 2 is the intersection of the neighborhood and the plane parallel to $\Theta_2$ shown in the figure, while the distribution is along the $x_2$-axis. For agent 1, the foliation $S_1$ is the level set (independent of $x_2$) and the distribution $D_2$ is the corresponding tangent. Then, the overall foliation $S$ and distribution $D$ are given by $S=S_1\times S_2$ and $D=D_1\times D_2$. It is shown in \cite{} that, in a sufficiently small neighborhood, the relationships among $S$, $D $ and $F$ are still the same as the single-agent case. The messages of the agents are the level set indices. When designing the mechanism, an effective approach is to find the distribution by differentiating the goal function and then obtain the foliation by integrating the distribution. The messages are then obtained from the indices of the foliations.

\begin{figure}
  \centering
  \includegraphics[scale=0.3]{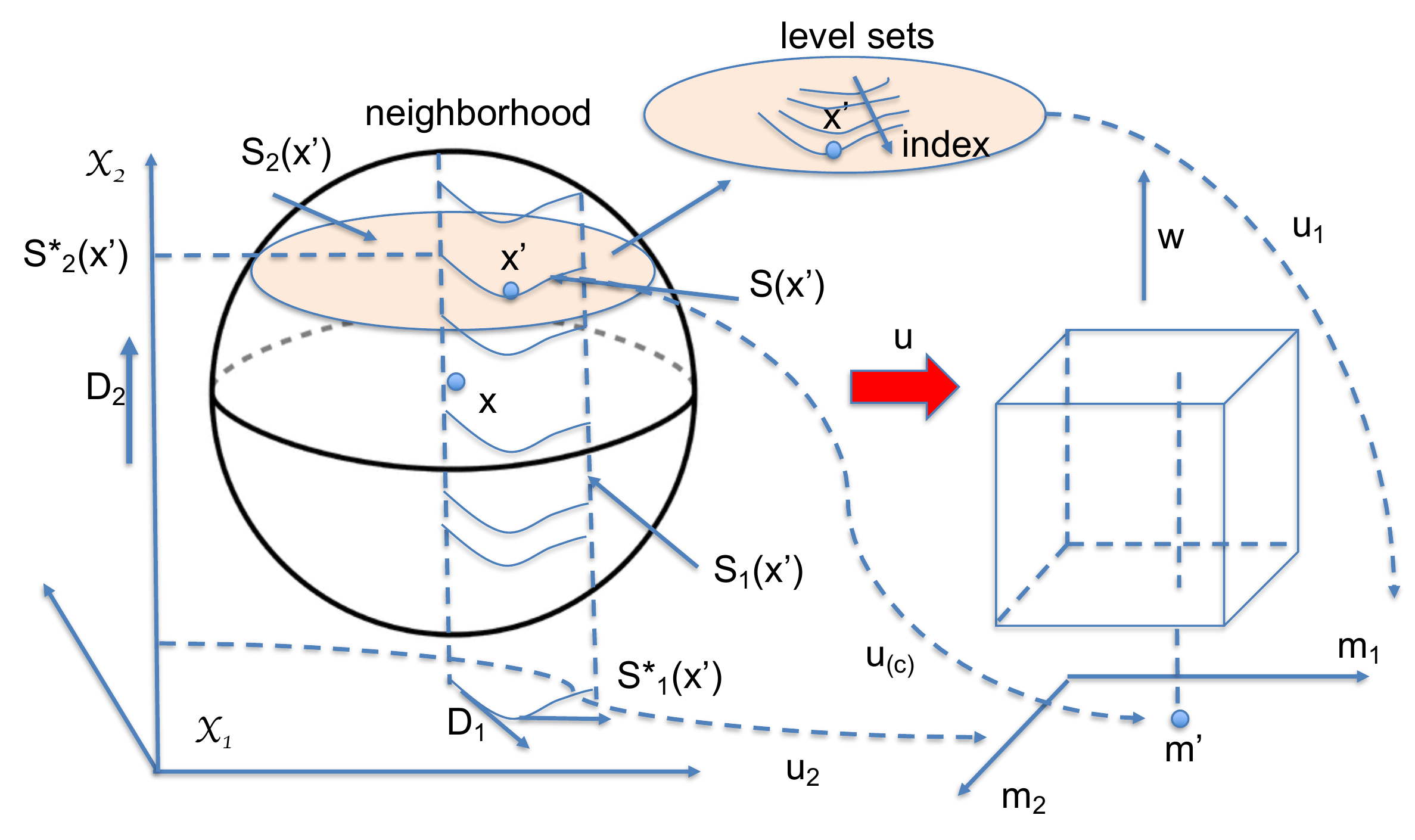}
  \caption{An illustration of the mechanism.}\label{fig:Fro_mech}
  \vspace{-0.1in}
\end{figure}

For the general case, the level set of agent $i$ is denoted by $S_i$ with dimensional $c_i+d_i$, $1\leq i\leq N$, while the corresponding distributions are denoted by $D_1$, ..., $D_n$. The overall level set is the product of the individual level sets, namely
\begin{eqnarray}
S(\theta)=\prod_{i=1}^N S_i(\theta),
\end{eqnarray}
while the overall distribution is the direction sum of the individual ones:
\begin{eqnarray}
D(\theta)=\oplus_{i=1}^ND_i(\theta).
\end{eqnarray}

The following theorem (the mechanism design version of Frobenius Theorem) discloses the geometric structure of mechanism design:
\begin{theorem}\label{thm:Frobenius}
Let $D_1$, ..., $D_n$ be $C^{\infty}$ distributions on $\Theta$. Let $\theta^*$ be any point in $\Theta$. There exists an open neighborhood $O(\theta^*)\subset \Theta$ of $\theta^*$, a local coordinate system $u:O(\theta^*)\rightarrow (-\epsilon,\epsilon)^{c+d}$, and an inverse mapping $v=u^{-1}:(-\epsilon,\epsilon)^{c+d}\rightarrow O(\theta^*)$ such that the following statement holds for any $\theta'\in O(\theta^*)$ and $(w',m')=u(\theta')$:
\begin{itemize}
\item A maximal, connected $d$-dimensional integral manifold $S(\theta')$ of $D$ exists through $\theta'$ in $O(\theta^*)$ and it satisfies
\begin{eqnarray}
S(\theta')&=&\{\theta\in O(\theta^*)|u_c(\theta)=u(\theta')=m'\}\nonumber\\
&=&\{\theta=v(w,m')|w\in (-\epsilon,\epsilon)^d\}
\end{eqnarray}

\item A maximal, connected $d_i$-dimensional integral manifold $S(\theta')$ of $D_i$ exists through $\theta'$ in $O(\theta^*)$ and it satisfies
\begin{eqnarray}
S_i(\theta')&=&\{(\theta_i,\theta_{-i}')\in O(\theta^*)|u(\theta)=u(\theta')=m'\}\nonumber\\
&=&\{\theta| u_c(\theta)=m',u_{j,d_j}(\theta)=w_j',j\neq i\}\nonumber\\
&=&\{\theta=v_i(w,m')|w_i\in (-\epsilon,\epsilon)^d\}
\end{eqnarray}

\item For each $i$, the rank of $D_{\theta_i}u_{i,c_i}$ has rank equal to $c_i$ on $O(\theta^*)$.
\item The mapping $v_i:(-\epsilon,\epsilon)^{c+d}\rightarrow\Theta_i$ depends only on the values of $w_i$ and $m$ and not on the value $w_{-i}$.
\end{itemize}
\end{theorem}
\begin{remark}
For each $i$, $c_i$ is the dimension of the useful information for the computing while $d_i$ is the dimension of the information that does not contribute to the computing. The total communication complexity, in terms of dimension, is $\sum_{i=1}^N c_i$. 
\end{remark}

\section{Mechanism Learning}
In traditional studies, the mechanism is designed using explicit analysis (e.g., linear programming \cite{} and level sets \cite{}) by human researchers. However, except for the simple case of two agents and a single goods \cite{}, the optimal solutions to most mechanism design problems have not been identified, probably due to the high complexity of problem. In the last two decades, there has been a trend to design the mechanism using machine learning, which takes the numerical methodology based on samples and is coined automated mechanism design (AMD). Thanks to the rapidly increasing computational capabilities of modern computers, there have been substantial breakthroughs in the area of AMD. Essentially, the AMD approach is to use sufficiently complex functions (e.g., SVM or deep neural network) to approximate the input-output relationship $F(\theta_1,...,\theta_n)$, while keeping the reports for the agents incentive compatible. The samples are obtained by randomly generating the private parameters $\{\theta_n\}_n$. The output of the learning procedure is the functions $g$ and $\{\psi_n\}_n$ in the given forms (e.g., neural network). The procedure is sketched in Fig. \ref{fig:learning}. Traditional studies on AMD include \cite{Conitzer2002,Conitzer2004,Guo2010,Sui2013} use heuristic searches. The deep learning approach is employed for AMD in \cite{Dutting2019}, while SVM is applied in \cite{Narasinimhan2016}. The sample complexity of AMD has been analyzed in \cite{Balcan2005}. However, such a learning methodology meets the following severe challenges in the context of spectrum markets for communications and sensing:
\begin{itemize}
\item Prior Distribution: In the spectrum market scenario, the major private parameters are the utility functions of the agents. Due to the complexity of function spaces, it is difficult to devise a good prior distribution for the utility functions. Parameterized functions with predetermined forms may not well generalize.
\item Offline Learning: The proposed mechanism learning algorithms are mostly offline, without considering the feedbacks of the agents during the operation. It may be more effective to learn the mechanism in an online manner, similarly to reinforcement learning. 
\item Black Box: Most existing mechanism learning algorithms are designed in a black box manner, without exploiting the intrinsic structure of the mechanism, which substantially decrease the efficiency of mechanism learning.  
\end{itemize}
Law enaction is a real world practice of mechanism design. A law is seldom set with offline and blackbox computations with artificial distributions. It needs to incorporate the understanding of laws (thus the structure) and be refined in the operation. Therefore, the PI plans to devise mechanism learning algorithms by exploiting online operation feedbacks and the geometric structures of mechanism, which will be elaborated in the details of Task 1. 

\section{Proof of Theorem \ref{thm:full_quadratic}}\label{appdx:proof_quadratic}
We consider $p_1$ and the case of scalar $A_1$ and $A_2$, which satisfy
\begin{eqnarray}\label{eq:scalar}
\left\{
\begin{array}{ll}
A_1u_1-\frac{\partial P}{\partial u_1}=0\\
A_2u_2-\frac{\partial P}{\partial u_2}=0
\end{array}
\right..
\end{eqnarray}

The Bordered Mixed Hessian matrix is given by
\begin{eqnarray}
BMH_{A_1,A_2}(p_1)=
\begin{pmatrix}
0 & \frac{\partial p_1}{\partial A_1}\\
\frac{\partial p_1}{\partial A_2} & \frac{\partial^2 p_1}{\partial A_1\partial A_2}
\end{pmatrix}
\end{eqnarray}

We simply need to verify whether $\frac{\partial p_1}{\partial A_1}=0$ or $\frac{\partial p_1}{\partial A_2}=0$. Recall that $p_1=A_1u_1$; therefore, we have $\frac{\partial p_1}{\partial A_1}=u_1+A_1\frac{\partial u_{1}}{\partial A_1}$ and $\frac{\partial p_1}{\partial A_2}=A_1\frac{\partial u_1}{\partial A_2}$. Therefore, we need to calculate $\frac{\partial u_{1}}{\partial A_1}$ and $\frac{\partial u_{1}}{\partial A_2}$.

Taking derivative with respect to $A_1$ on the first equation in (\ref{eq:scalar}), we obtain
\begin{eqnarray}
u_1+A_1\frac{\partial u_1}{\partial A_1}-\left(\frac{\partial u_1}{\partial A_1}\frac{\partial^2 P}{\partial u_1^2}+\frac{\partial u_2}{\partial A_1}\frac{\partial^2 P}{\partial u_1\partial u_2}\right)=0.
\end{eqnarray}

Taking derivative with respect to $A_2$ on the first equation in (\ref{eq:scalar}), we obtain
\begin{eqnarray}
A_1\frac{\partial u_1}{\partial A_2}-\left(\frac{\partial u_1}{\partial A_2}\frac{\partial^2 P}{\partial u_1^2}+\frac{\partial u_2}{\partial A_2}\frac{\partial^2 P}{\partial u_1\partial u_2}\right)=0.
\end{eqnarray}

Taking derivative with respect to $A_1$ on the second equation in (\ref{eq:scalar}), we obtain
\begin{eqnarray}
A_2\frac{\partial u_2}{\partial A_1}-\left(\frac{\partial u_1}{\partial A_1}\frac{\partial^2 P}{\partial u_1\partial u_2}+\frac{\partial u_2}{\partial A_1}\frac{\partial^2 P}{\partial u_2^2}\right)=0.
\end{eqnarray}

Taking derivative with respect to $A_2$ on the second equation in (\ref{eq:scalar}), we obtain
\begin{eqnarray}
u_2+A_2\frac{\partial u_2}{\partial A_2}-\left(\frac{\partial u_1}{\partial A_2}\frac{\partial^2 P}{\partial u_1\partial u_2}+\frac{\partial u_2}{\partial A_2}\frac{\partial^2 P}{\partial u_2^2}\right)=0.
\end{eqnarray}

Writing the above equations in the matrix form, we have
\begin{small}
\begin{eqnarray}\label{eq:matrix_form}
\begin{pmatrix}
A_1-\frac{\partial^2 P}{\partial u_1^2} & 0 & -\frac{\partial^2 P}{\partial u_1\partial u_2} & 0\\
0 & A_1- \frac{\partial^2 P}{\partial u_1^2} & 0 & -\frac{\partial^2 P}{\partial u_1\partial u_2}\\
-\frac{\partial^2 P}{\partial u_1\partial u_2} & 0 & A_2-  \frac{\partial^2 P}{\partial u_2^2} & 0\\
0 &  -\frac{\partial^2 P}{\partial u_1\partial u_2} & 0 & A_2- \frac{\partial^2 P}{\partial u_2^2}
\end{pmatrix}
\begin{pmatrix}
\frac{\partial u_1}{\partial A_1}\\
\frac{\partial u_1}{\partial A_2}\\
\frac{\partial u_2}{\partial A_1}\\
\frac{\partial u_2}{\partial A_2}
\end{pmatrix}
=
\begin{pmatrix}
-u_1\\
0\\
0\\
-u_2
\end{pmatrix}
\end{eqnarray}
\end{small}

\begin{itemize}
\item If $\frac{\partial p_1}{\partial A_2}=0$, we have $\frac{\partial u_1}{\partial A_2}=0$. From the second row of (\ref{eq:matrix_form}), we have $\frac{\partial u_2}{\partial A_2}=0$, which conflict with the fourth row of (\ref{eq:matrix_form}).

\item if $\frac{\partial p_1}{\partial A_2}=0$, we have $\frac{\partial u_1}{\partial A_1}=-\frac{u_1}{A_1}$. From the first row of (\ref{eq:matrix_form}), we have $\frac{\partial u_1}{\partial A_2}=-\frac{\frac{ u_1}{ A_1}\frac{\partial^2 P}{\partial u_1^2}}{\frac{\partial^2P}{\partial u_1\partial u_2}}$. Substituting it into the third row of (\ref{eq:matrix_form}), we have
\begin{eqnarray}
\left(\frac{\partial^2 P}{\partial u_1\partial u_2}\right)^2=A_2\frac{\partial^2 P}{\partial u_1^2}-\left(\frac{\partial^2 P}{\partial u_1^2}\right)^2
\end{eqnarray}
\end{itemize}

An alternative approach: When the dimension can be reduced, for any $\mathbf{p}_1$ and $\mathbf{p}_2$, for any $\mathbf{A}_1$, the level set
\begin{eqnarray}
\mathcal{A}_2(\mathbf{p}_1,\mathbf{p}_2,\mathbf{A}_1)=\left\{\mathbf{A}_2|F(\mathbf{A}_1,\mathbf{A}_2)=(\mathbf{p}_1,\mathbf{p}_2)|\mathbf{A}_1\right\}.
\end{eqnarray}
is independent of $\mathbf{A}_1$. From the first equation of (\ref{eq:scalar}), $\mathbf{u}_2$ is also uniquely determined (?). Hence, 
\begin{eqnarray}
\mathcal{A}_2=\{\mathbf{A}_2 \mathbf{u}_2=\mathbf{p_2}\}.
\end{eqnarray}
Since $\mathbf{A}_1\mathbf{u}_1=\mathbf{p_1}$ and $\mathbf{A}_1$ is full rank, $\mathbf{u}_1$ is fixed. Meanwhile, we have $\mathbf{A}_2\mathbf{u}_2=\mathbf{p_1}$.

Now, we take an alternative $\mathbf{A}_1'$ such that 
\begin{eqnarray}
\mathbf{u}_1'={\mathbf{A}_1'}^{-1}\mathbf{p}_1\neq \mathbf{u}_1.
\end{eqnarray}
The reason for the existence of $\mathbf{A}_1'$ is: suppose that there is no such an $\mathbf{A}_1'$, we have $(\mathbf{A}_1-\mathbf{A}_1')\mathbf{u}_1=0$ for all $\mathbf{A}'_1$, which is impossible when $\mathbf{u}_1\neq 0$. Then, the different $\mathbf{u}_1'$ results in different $\mathbf{u}_2'$. The corresponding level set is given by
\begin{eqnarray}
\mathcal{A}_2'=\{\mathbf{A}_2 \mathbf{u}_2'=\mathbf{p_2}\}.
\end{eqnarray}
If $\mathcal{A}_2'=\mathbf{A}_2$, we have
\begin{eqnarray}
\mathbf{A}_2 (\mathbf{u}_2-\mathbf{u}_2')=0,
\end{eqnarray}
for all $\mathbf{A}_2\in \mathcal{A}_2'$, where $\mathbf{u}_2-\mathbf{u}_2'\neq 0$. This is impossible for any full rank $\mathbf{A}_2$. 

\section{Proof of Theorem \ref{thm:simultaneous}}\label{appdx:simultaneous}
\begin{proof}
We first assume that the higher order derivatives are zero. Then, the difference between the object function $U_n+R$ at system $n$ is given by
\begin{eqnarray}
&&U_n(\mathbf{u}_n^{t})+R(\mathbf{u}_n^{t},\mathbf{u}_{-n}^{t-1})-U_n(\mathbf{u}_n^{t-1})-R(\mathbf{u}_n^{t-1},\mathbf{u}_{-n}^{t-1})\nonumber\\
&=&\mathbf{a}_n\delta \mathbf{u}_n+\delta\mathbf{u}_n^T\mathbf{H}_n\delta\mathbf{u}_n.
\end{eqnarray}
where $\mathbf{a}=\nabla_{\mathbf{u}_n}(U_n+R)$ is the gradient vector and $\mathbf{H}$ is the Hessian matrix of $R$ as a function of $\mathbf{u}_n$. Since system $n$ maximizes the increase of the local object function, it should select
\begin{eqnarray}
\delta \mathbf{u}_n=-\frac{1}{2}\mathbf{H}_n^{-1}\mathbf{a}_n,
\end{eqnarray}
which implies
\begin{eqnarray}
&&U_n(\mathbf{u}_n^{t})+R(\mathbf{u}_n^{t},\mathbf{u}_{-n}^{t-1})-U_n(\mathbf{u}_n^{t-1})-R(\mathbf{u}_n^{t-1},\mathbf{u}_{-n}^{t-1})\nonumber\\
&=&\frac{1}{2}\mathbf{a}^T\mathbf{H}_n^{-1}\mathbf{a}_n.
\end{eqnarray}

Then, the difference of the social welfares at stages $t$ and $t-1$ is given by
\begin{eqnarray}
&&W(\mathbf{u}^{t})-W(\mathbf{u}^{t-1})\nonumber\\
&=&\frac{1}{2}\sum_{n=1}^N\mathbf{a}^T\mathbf{H}_n^{-1}\mathbf{a}_n
+\sum_{n\neq m}\delta \mathbf{u}_{n}^T\mathbf{H}_{nm}\delta \mathbf{u}_m\nonumber\\
&=&\frac{1}{2}\sum_{n=1}^N\delta \mathbf{u}_n^T\mathbf{H}_n\delta\mathbf{u}_n
+\sum_{n\neq m}\delta \mathbf{u}_{n}^T\mathbf{H}_{nm}\delta \mathbf{u}_m
\end{eqnarray}

When the elements in $\mathbf{H}_{nm}$ are sufficiently small, the above difference is positive. Since the higher order derivatives are also sufficiently small, the difference is still positive, even if the higher order terms are taken into account.

Therefore, the objective function value always increases. Since it is bounded, it converges. This concludes the proof.

\end{proof}

\section{Proof of Theorem ?}
First, we need the following lemma (Lemma 12.1.5. \cite{Facchinei2003}):
\begin{lemma}\label{lem:converge}
Suppose that a series of functions $\{F^k\}_k$ are co-coercive with constants $\{c_k\}_k$, which satisfies
\begin{eqnarray}
\rho=\inf_k c_k>\frac{1}{2}.
\end{eqnarray}
If all the functions $F^k$ have the same nonempty set $S$ of zeros and
\begin{eqnarray}
\rho=\inf_k c_k>\frac{1}{2}.
\end{eqnarray}
and
\begin{eqnarray}
\inf_k\|F^k(x)\|>0,\qquad \forall x\in S,
\end{eqnarray}
then the sequence $x^k$ produced by $x^{k}=x^{k-1}+F^k(x^{k-1})$ converges to a point $x^*$ in S. 
\end{lemma}

Then, we can prove the main theorem.
\begin{proof}
We define $F^k(\mathbf{u})=\tau_k(\nabla_{\mathbf{u}_1}R_1(\mathbf{u}),...,\nabla_{\mathbf{u}_1}R_1(\mathbf{u}))$, and
\begin{eqnarray}
\tilde{F}(\mathbf{u})=F(\mathbf{u}'),
\end{eqnarray}
where $\mathbf{u}'$ is obtained from (\ref{eq:stage1}), namely
\begin{eqnarray}
&&\nabla_{\mathbf{u}_n} U_n(\mathbf{u}_n')+\nabla_{\mathbf{u}_n}\Psi(\mathbf{u}_n',\mathbf{u}_{-n})\nonumber\\
&-&\lambda_k(\mathbf{u}_n'-\mathbf{u}_n)=0.
\end{eqnarray}

Then, for two points $\mathbf{u}$ and $\mathbf{v}$ in $\mathbb{R}^d$, we have
\begin{eqnarray}
&&(\tilde{F}^k(\mathbf{u})-\tilde{F}^k(\mathbf{v}))^T(\mathbf{u}-\mathbf{v})\nonumber\\
&=&(F^k(\mathbf{u})-F^k(\mathbf{v}))^T(\mathbf{u}-\mathbf{v})\nonumber\\
&+&(\mathbf{e}^k(\mathbf{u})-\mathbf{e}^k(\mathbf{v}))^T(\mathbf{u}-\mathbf{v}).
\end{eqnarray}
The first term satisfies 
\begin{eqnarray}
(F^k(\mathbf{u})-F^k(\mathbf{v}))^T(\mathbf{x}-\mathbf{y})\geq c_k^1\|F^k(\mathbf{u})-F^k(\mathbf{v})\|_2^2,
\end{eqnarray}
and the second term satisfies
\begin{eqnarray}
(\mathbf{e}^k(\mathbf{u})-\mathbf{e}^k(\mathbf{v}))^T(\mathbf{x}-\mathbf{y})\geq c_k^2\|\mathbf{e}^k(\mathbf{u})-\mathbf{e}^k(\mathbf{v})\|_2^2,
\end{eqnarray}
due to the assumption.

Meanwhile, we have 
\begin{eqnarray}
\|\tilde{F}^k(\mathbf{u})-\tilde{F}^k(\mathbf{v})\|_2^2&\leq& \|F^k(\mathbf{u})-F^k(\mathbf{v})\|_2^2\nonumber\\
&+&\|\mathbf{e}^k(\mathbf{u})-\mathbf{e}^k(\mathbf{v})\|_2^2.
\end{eqnarray}

Therefore, we have 
\begin{eqnarray}
(\tilde{F}^k(\mathbf{u})-\tilde{F}^k(\mathbf{v}))^T(\mathbf{x}-\mathbf{y})\geq \min\{c_k^1,c_k^2\}\|\tilde{F}^k(\mathbf{u})-\tilde{F}^k(\mathbf{v})\|_2^2.
\end{eqnarray}

This satisfies the condition of Lemma \ref{lem:converge}. The subsequent proof follows that of Theorem 12.1.8 in \cite{Facchinei2003}.

\end{proof}

\section{Proof of Theorem ?}
We denote by $\mathbf{u}^*(\mathbf{x})$ the optimal control action $\mathbf{u}$ at state $\mathbf{x}$.
We assume that the convergence of the algorithm with fixed $\mathbf{x}_)$ and initial value $\mathbf{u}_0$ satisfies
\begin{eqnarray}\label{eq:decay}
\|\mathbf{u}^k-\mathbf{u}^*(\mathbf{x}_0)\|\leq E_k\|\mathbf{u}_0-\mathbf{u}^*(\mathbf{x}_0)\|,
\end{eqnarray}
where $E_k$ decreases with $k$ and is independent of $\mathbf{x}_0$ and $\mathbf{u}_0$. For the two-stage and single-stage algorithms, we can write the dynamics as
\begin{eqnarray}
\mathbf{u}^k=\mathbf{u}^{k-1}+\tilde{F}(\mathbf{u}^{k-1},\mathbf{x}^{k-1}).
\end{eqnarray}

\begin{theorem}
Denote by $\lambda_B$ the maximum absolute value of the eigenvalues of $\mathbf{B}$. Assume that $\|\mathbf{u}\|\leq u_m$. We also define $h_{\tilde{F}}=\sup_{\mathbf{u},\mathbf{x}}\|H_{\tilde{F}}(\mathbf{u},\mathbf{x})\|_2$, where $H_{\tilde{F}}$ is the Hessian matrix of $\tilde{F}$. Then, we have
\begin{eqnarray}
\|\mathbf{u}^k-\mathbf{u}^*(\mathbf{x}^k)\|&\leq& \min_n\left\{\frac{h_{\tilde{F}}\lambda_Bnu_m}{1-\lambda_B}\right.\nonumber\\
&-&\left.\frac{h_{\tilde{F}}\lambda_B(1-\lambda_B^n)u_m}{(1-\lambda_B)^2}+2E_nu_m\right\}.
\end{eqnarray}
\end{theorem}

\begin{proof}
Fixing an $n$, we rewrite the error $\mathbf{u}^k-\mathbf{u}^*(\mathbf{x})$ as 
\begin{eqnarray}\label{eq:decompose}
\mathbf{u}^k-\mathbf{u}^*(\mathbf{x}^k)&=&\mathbf{u}^k-\bar{\mathbf{u}}^n(\mathbf{x}^k,\mathbf{u}^{k-n})\nonumber\\
&+&\bar{\mathbf{u}}^n(\mathbf{x}^k,\mathbf{u}^{k-n})-\mathbf{u}^*(\mathbf{x}^k),
\end{eqnarray}
where $\bar{\mathbf{u}}^n(\mathbf{x},\mathbf{u}^{k-n})$ is the outcome of applying $\tilde{F}$ for $n$ times, beginning from $\mathbf{u}^{k-n}$ and fixing $\mathbf{x}^k$. 

According to (\ref{eq:decay}), the second term of (\ref{eq:decompose}) is bounded by
\begin{eqnarray}
\|\bar{\mathbf{u}}^n(\mathbf{x}^k,\mathbf{u}^{k-n})-\mathbf{u}^*(\mathbf{x}^k)\|&\leq& E_n\|\mathbf{u}^{k-n}-\mathbf{u}^*(\mathbf{x}^k)\|\nonumber\\
&\leq& 2E_nu_m.
\end{eqnarray}

For the first term, we define 
\begin{eqnarray}
\delta \mathbf{u}^l=\mathbf{u}^l-\mathbf{u}^{l-1}=\tilde{F}(\mathbf{u}^{l-1},\mathbf{x}^{l-1}),
\end{eqnarray}
and 
\begin{eqnarray}
\delta \bar{\mathbf{u}}^l=\bar{\mathbf{u}}^l-\bar{\mathbf{u}}^{l-1}=\tilde{F}(\bar{\mathbf{u}}^{l-1},\mathbf{x}^{k}),
\end{eqnarray}

Then, we have 
\begin{eqnarray}
&&\|\mathbf{u}^k-\bar{\mathbf{u}}^n(\mathbf{x}^k,\mathbf{u}^{k-n})\|\nonumber\\
&=&\left\|\sum_{p=0}^{n-1} \delta \mathbf{u}^{k-p}-\delta\bar{\mathbf{u}^{k-p}}\right\|\nonumber\\
&\leq &\sum_{p=0}^{n-1}\left\| \delta \mathbf{u}^{k-p}-\delta\bar{\mathbf{u}^{k-p}}\right\|\nonumber\\
&\leq & \sum_{p=0}^{n-1}\mathbf{H}_{\tilde{F}} \|\mathbf{x}^{k-p}-\mathbf{x}^k\|  \nonumber\\
&\leq & \sum_{p=0}^{n-1}h_{\tilde{F}}\frac{\lambda_B(1-\lambda_B^p)}{1-\lambda_B}\nonumber\\
&=&\frac{h_{\tilde{F}}\lambda_Bnu_m}{1-\lambda_B}-\frac{h_{\tilde{F}}\lambda_B(1-\lambda_B^n)u_m}{(1-\lambda_B)^2}.
\end{eqnarray}
where we used the fact that
\begin{eqnarray}
\|\mathbf{x}^{k-p}-\mathbf{x}^k\|\leq \sum_{q=1}^p \lambda_B^qu_m=\frac{\lambda_B(1-\lambda_B^p)}{1-\lambda_B}.
\end{eqnarray}

This concludes the proof by choosing the $n$ that minimizes the upper bound. 
\end{proof}

\section{Proof of Theorem \ref{thm:infinity}}
\begin{theorem}\label{thm:infinity}
The dimension of the message space for the computing task in (?) is infinite, if $\Theta$ cannot be written as 
\begin{eqnarray}\label{eq:decompose}
\Theta(\mathbf{u}_1,...,\mathbf{u_N})=\phi_1(\mathbf{u}_1)+...+\phi_1(\mathbf{u}_N).
\end{eqnarray}
\end{theorem}

\begin{proof}
For notational simplicity, we assume $N=2$ and $d=1$, namely the two-subsystem scalar dynamics, which does not lose the generality. We consider only the computation of $u_1$, such that the output function $F$ is scalar. Moreover, we assume that $U_1(0)$ is a constant for all possible $U_1'$, and so is $U_2$. Then, $U_1$ and $U_1+C$ ($C\neq 0$) are considered to be the same function. 

Suppose that (\ref{eq:decompose}) does not hold while the message space dimension is finite. Our goal is to prove that the dimension $c_1$ of useful information, for the purpose of computing, cannot be bounded. To this end, we assume that $c_1,c_2<\infty$. According to Theorem \ref{thm:Frobenius}, the level set $S_i(\theta')$ is independent of $w_2$, while $d_2>0$. We consider $(U_1,U_2)$ such that the solution is $(u_1^*,u_2^*)$. For any $u_2'$, we can find a $U_2'$ such that the output is $(u_1^*,u_2)$. Due to the independence of $S_1$ on $U_2$, we have $S_1(U_1,U_2)=S_1(U_1,U_2')$ as the level set of $u_1^*$. According to the equation
\begin{eqnarray}
\frac{\partial U_1}{\partial u_1}\bigg|_{u_1^*}= -\frac{\partial \Theta}{\partial u_1}\bigg|_{u_1^*,u_2'},
\end{eqnarray}
and the arbitrary values $u_1^*$ and $u_2'$, we obtain that $\frac{\partial \Theta}{\partial u_1}$ is independent of $u_2$. Then, we define 
\begin{eqnarray}
\phi_1(u_1)=W(0,0)+\int_0^{u_1}\frac{\partial \Theta(v_1,0)}{\partial v_1}dv_1,
\end{eqnarray}
and further
\begin{eqnarray}
\phi_2(u_1,u_2)=\Theta(u_1,u_2)-\phi_1(u_1).
\end{eqnarray}
It is easy to show that
\begin{eqnarray}
\frac{\phi_2(u_1,u_2)}{\partial u_1}=0.
\end{eqnarray}
Therefore, $\phi_2(u_1,u_2)=\phi_2(u_2)$. In this way, we have 
\begin{eqnarray}
\Theta(u_1,u_2)=\phi_1(u_1)+\phi_2(u_2).
\end{eqnarray}

This concludes the proof by the contradiction.

\end{proof}

\end{document}